\pdfoutput=1
\documentclass[11pt]{article}

\usepackage{array}

\usepackage[letterpaper,margin=1in]{geometry}
\usepackage{lipsum}

\newcommand\blfootnote[1]{%
  \begingroup
  \renewcommand\thefootnote{}\footnote{#1}%
  \addtocounter{footnote}{-1}%
  \endgroup
}

\usepackage[utf8]{inputenc} %
\usepackage[T1]{fontenc}    %

\usepackage{url}            %
\usepackage{booktabs}       %
\usepackage{nicefrac}       %
\usepackage{microtype}      %
\usepackage{enumitem}
\usepackage{dsfont}
\usepackage{multicol}
\usepackage{xcolor}
\usepackage{mdframed}
\usepackage{multirow}

\usepackage{amsthm,amsfonts,amsmath,amssymb,epsfig,color,float,graphicx,verbatim, enumitem}

\usepackage{algorithmicx}
\usepackage[noend]{algpseudocode}
\usepackage{algorithm}

\usepackage{mathtools}
\usepackage{bbm}

\usepackage{thm-restate}

\definecolor{green}{rgb}{0.0, 0.5, 0.0}
\usepackage[colorlinks,citecolor=blue,linkcolor=magenta,bookmarks=true,hypertexnames=false]{hyperref}
\usepackage[nameinlink,capitalize]{cleveref}

\crefformat{equation}{(#2#1#3)}
\crefname{lemma}{lemma}{lemmata}
\crefname{claim}{claim}{claims}
\crefname{theorem}{theorem}{theorems}
\crefname{proposition}{proposition}{propositions}
\crefname{corollary}{corollary}{corollaries}
\crefname{claim}{claim}{claims}
\crefname{remark}{remark}{remarks}
\crefname{definition}{definition}{definitions}
\crefname{problem}{problem}{problems}
\crefname{fact}{fact}{facts}
\crefname{question}{question}{questions}
\crefname{condition}{condition}{conditions}
\crefname{algorithm}{algorithm}{algorithms}
\crefname{assumption}{assumption}{assumptions}
\crefname{notation}{notation}{notation}
\crefname{cond}{Condition}{Conditions}

  {%
   \par\noindent{\bfseries\upshape Proof Sketch\ }%
  }%
  {\qed}

\newtheorem{theorem}{Theorem}[section]
\newtheorem{lemma}[theorem]{Lemma}

\newtheorem{claim}[theorem]{Claim}

\newtheorem{definition}[theorem]{Definition}

\newtheorem{fact}[theorem]{Fact}

\theoremstyle{definition}

\newcommand{\eps}{\varepsilon}

\newcommand{\dtv}{D_\mathrm{TV}}
\newcommand{\Ind}{\mathds{1}}
\newcommand{\1}{\Ind}

\renewcommand{\Pr}{\operatorname*{\mathbb{P}}}

\newcommand{\E}{\operatorname*{\mathbb{E}}}
\newcommand{\from}{\leftarrow}

\newcommand{\poly}{\operatorname*{\mathrm{poly}}}

\newcommand{\DTV}{D_\mathrm{TV}}
\newcommand{\DKL}{D_\mathrm{KL}}

\renewcommand{\d}{\mathrm{d}}
\renewcommand{\tilde}{\widetilde}
\renewcommand{\hat}{\widehat}

\def\R{\mathbb R}

\def\Z{\mathbb Z}

\newcommand{\cA}{\mathcal{A}}

\newcommand{\cE}{\mathcal{E}}

\newcommand{\cN}{\mathcal{N}}

\newcommand{\cX}{\mathcal{X}}

\newcommand\snorm[2]{\left\| #2 \right\|_{#1}}

\def\d{\mathrm{d}}

\def\colorful{0}

\ifnum\colorful=1
\newcommand{\inote}[1]{\footnote{{\bf [Ilias: {#1}\bf ] }}}
\newcommand{\dnote}[1]{\footnote{{\bf [Daniel: {#1}\bf ] }}}
\newcommand{\anote}[1]{\footnote{{\bf [Ankit: {#1}\bf ]}}}
\newcommand{\tnote}[1]{\footnote{{\bf [Thanasis: {#1}\bf ] }}}
\newcommand{\jd}[1]{\footnote{{\bf [Jelena: {#1}\bf ] }}}

\else

\newcommand{\inote}[1]{}
\newcommand{\dnote}[1]{}
\newcommand{\anote}[1]{}
\newcommand{\tnote}[1]{}
\newcommand{\jd}[1]{}

\fi

\allowdisplaybreaks

\title{Linear Regression under Missing or Corrupted Coordinates}

\author{
Ilias Diakonikolas\thanks{Supported by NSF Medium Award CCF-2107079, ONR award number N00014-25-1-2268, 
and an H.I. Romnes Faculty Fellowship.}\\
University of Wisconsin-Madison\\
{\tt ilias@cs.wisc.edu}\\
\and
Jelena Diakonikolas\thanks{Supported in part by the Air Force Office of Scientific Research under award number FA9550-24-1-0076, by the U.S. Office of Naval Research under contract number N00014-22-1-2348, and by the NSF CAREER Award CCF-2440563. Any opinions, findings and conclusions or recommendations expressed in this material are those of the author(s) and do not necessarily reflect the views of the U.S. Department of Defense.}\\
University of Wisconsin-Madison\\
{\tt jelena@cs.wisc.edu}\\
\and
Daniel M. Kane\thanks{Supported by NSF Medium Award CCF-2107547 and NSF Award CCF-1553288 (CAREER).}\\
University of California, San Diego\\
{\tt dakane@cs.ucsd.edu}
\and
Jasper C.H. Lee\thanks{Supported in part by NSF Medium Award CCF-2107079, NSF AiTF Award CCF-2006206, and a Croucher Fellowship for Postdoctoral Research.}\\
University of California, Davis\\
{\tt jasperlee@ucdavis.edu}\\
\and
Thanasis Pittas\thanks{Supported by NSF Medium Award CCF-2107079 and NSF Award DMS-2023239 (TRIPODS).}\\
University of Wisconsin-Madison\\
{\tt pittas@wisc.edu}\\
}
\date{\today}

\begin{document}

\maketitle

\begin{abstract}%
We study multivariate linear regression under Gaussian covariates in two settings, where data may be erased or corrupted by an adversary under a coordinate-wise budget.
In the incomplete data setting, an adversary may inspect the dataset and delete entries in up to an $\eta$-fraction of samples per coordinate—a strong form of the Missing Not At Random model.
In the corrupted data setting, the adversary instead replaces values arbitrarily, and the corruption locations are unknown to the learner.
Despite substantial work on missing data, linear regression under such adversarial missingness remains poorly understood, even information-theoretically.
Unlike the clean setting, where estimation error vanishes with more samples, here the optimal error remains a positive function of the problem parameters.
Our main contribution is to characterize this error up to constant factors across essentially the entire parameter range. Specifically, we establish novel information-theoretic lower bounds on the achievable error that match the error of (computationally efficient) algorithms.
A key implication is that, perhaps surprisingly, the optimal error in the missing data setting matches the error in the corruption setting—so knowing the corruption locations offers no general advantage.
\end{abstract}
\blfootnote{Authors are in alphabetical order.}

 \setcounter{page}{0}

\thispagestyle{empty}

\newpage

\section{Introduction}\label{sec:intro}

We study regression tasks in settings that deviate from the ideal i.i.d.~assumption, motivated by the prevalence of imperfect or incomplete data in modern machine learning. Missing values can arise from sensor failures, survey nonresponse, or data aggregation processes such as crowdsourcing \cite{vuurens2011much} and peer grading \cite{PHC+2013,KWL+2013}. Moreover, datasets may include corrupted entries, which may stem from data poisoning attacks \cite{BarNJT10-security,BigNL12-poisoning,SteKL17-certified,TraLM18-backdoor,HayKSO21}, out-of-distribution examples \cite{YZLL2024}, or biological anomalies \cite{RosPWCKZF02-Gen,PasLJD10-ancestry,LiATSCRCBFCM08-science}.
These two types of impurities (incompleteness and corruption) have given rise to two largely distinct lines of research:  \emph{learning with missing data}  and  \emph{robust statistics}.\looseness=-1

\paragraph{Incomplete data} Rubin's seminal work \cite{Rub1976a} classifies missing data mechanisms based on their dependence on observed values: if missingness is independent of the data, it is Missing Completely At Random (MCAR); if it depends only on observed values, it is Missing At Random (MAR); otherwise, it is Missing Not At Random (MNAR).
A common approach to handling missing data is \emph{imputation} of missing entries \cite{JCP+2024,LJSV2021,ABDS2023}, though not all methods rely on it. Early work in the area focused on parameter estimation \cite{Lit1992,Lit1993,LW2012,RT2010}, whereas more recent research has shifted toward prediction \cite{MPJ+2020,LJM+2020,JCP+2024}, a more challenging task since missing features during test time can prevent model evaluation. \looseness=-1

\paragraph{Corrupted data} The robust statistics, initiated in the 1960s \cite{Tuk60, Hub64}, focuses on parameter estimation from fully observed data with a small fraction of arbitrary corruptions. A recent resurgence in this area \cite{LRV2016,DKK+19} has extended classical one-dimensional results to high dimensions while preserving sample and computational efficiency (see \cite{DK2023}).
While this literature can also handle missing data by na\"ively imputing erased locations with arbitrary values, this simple strategy ignores the fact that the erasure locations are known and can potentially be leveraged by a more sophisticated algorithm. Additional related work on missing data and robust statistics is given in \Cref{app:related_work}. \looseness=-1

In this work, we study the relationship between these two types of data impurities in the context of parameter estimation in Gaussian linear regression.
\begin{definition}[Linear regression]\label{def:linear_model}
A labeled example $(X, y)$ follows the linear regression model with regressor $\beta$ and additive noise level $\sigma$ if $X \sim \cN(0, I)$ and $y = \beta^\top X + \xi$, where $\xi \sim \cN(0, \sigma^2)$.
\end{definition}
\vspace{-3pt}

We model data impurities via a coordinate-wise adversary: for each coordinate\footnote{The corruption budget is applied per coordinate rather than globally, to prevent cases where $\beta$ has a single non-zero entry that receives all corruptions.}, at most an $\eta$-fraction of samples may be modified by erasure or replacement.
Unlike much of the robust statistics literature—which treats each sample as fully clean or corrupted—partial corruptions are common, especially when coordinates correspond to different sensors or measurement methods \cite{TCS+2001}.
\Cref{def:model} is a strong MNAR variant, as missing-entry locations are chosen adversarially based on the dataset.\looseness=-1

\begin{restatable}[Coordinate-wise data incompleteness]{definition}{CORRUPTIONS}\label{def:model}
We define the coordinate-wise incomplete version $\tilde S$ of a set $S=\{ (X^{(i)}, y^{(i)}) \}_{i \in [n]}$ of $n$ labeled examples, where $X^{(i)} \in \R^d$ and $y^{(i)} \in \R$, as follows:
An adversary inspects $S$ and, for each coordinate $j \in [d]$, is allowed to erase the $j$-th coordinate of up to $\eta n$ of the $X^{(i)}$'s, replacing them with the special symbol $\perp$ to yield $\{\tilde{X}^{(i)}\}$. Similarly, it can replace up to an $\eta$-fraction of the $y^{(i)}$'s with $\perp$ to yield $\{\tilde{y}^{(i)}\}$. Then, $\tilde S = \{ (\tilde X^{(i)}, \tilde y^{(i)}) \}_{i \in [n]}$.\looseness=-1
\end{restatable}

\begin{restatable}[Coordinate-wise data corruption]{definition}{CORRUPTIONSREPLACED}\label{def:model2}
As in \Cref{def:model}, but instead of erasing coordinates (replacing them with $\perp$), the adversary may substitute arbitrary values.
Crucially, the locations of the modified values are not known to the algorithm.
\end{restatable}

The most relevant prior work on \Cref{def:model,def:model2} is \cite{LPRT2021}, which focuses on mean estimation. While linear regression has been studied under various missing data settings \cite{LW2012,CAM2020,CCD2023}, existing work typically assumes much milder missingness than the fully adversarial pattern in \Cref{def:model}. To the best of our knowledge, even for the most ``textbook`` and standard setting of Gaussian covariates used in \Cref{def:model}, linear regression under the aforementioned corruption models remains unexplored, and our understanding is limited even from an information-theoretic standpoint. Specifically, we ask: 
\begin{center}
\emph{What is the optimal estimation error achievable with unlimited samples and computational power?}
\end{center}
Due to the adversarial nature of \Cref{def:model}, the error need not vanish as the number of samples grows. Instead, the best achievable error is a positive function of $\eta, \sigma, d, \beta$ that we seek to characterize. We provide matching upper and lower bounds, with the upper bounds attained by algorithms that are both sample and computationally efficient.\looseness=-1

Another question arises from comparing \Cref{def:model,def:model2}. When coordinates are missing rather than replaced, their locations are visible and may be leveraged for improved error, thus:\looseness=-1
\vspace{-3pt}
\begin{center}
\emph{Is the optimal estimation error under missing data smaller than the error under replaced data?}
\end{center}
\vspace{-3pt}
Prior algorithms for related problems with missing data \cite{LPRT2021,HR2009}, like mean estimation under \Cref{def:model} and low-rank structure \cite{LPRT2021}, rely heavily on knowing corruption locations and only handle special cases under replacement corruption.
Surprisingly, for linear regression, the answer to the above question is negative.\looseness=-1

\subsection{Our Results}\label{sec:results}

We now present the upper and lower bounds that address the two stated questions. The upper bounds follow  from existing computationally-efficient estimators, while the lower bounds are the main contributions of this work.\looseness=-1

We start with three estimators, all applicable to the stronger corruption model of \Cref{def:model2} (i.e., replaced data) and efficient in runtime and sample complexity. The first, $\cA_1$, is simply using an estimator from existing robust statistics literature, designed for the setting where $\eps$-samples are (entirely) corrupted with $\eps=\eta d$. The second, $\cA_2$, is based on the fact that under \Cref{def:linear_model}, $\E[yX] = \beta$, thus one can estimate each coordinate of $\beta$ using existing 1D robust Gaussian mean estimators tolerant to $2\eta$ corruption (like trimmed mean or median). The third, $\cA_3$, is the trivial algorithm returning always the zero vector. See \Cref{app:intro} for details on these algorithms.
As a note, $\cA_1,\cA_2,\cA_3$ are known to generalize naturally to other covariate distributions beyond Gaussians, but for concreteness and for clear comparison with our lower bounds, we focus on the most basic setting stated below.

\begin{restatable}[Upper bounds on estimation error]{fact}{KNOWNALGORITHMS}\label{fact:known-algorithms}
    There are three polynomial-time algorithms $\cA_1,\cA_2,\cA_3$ that take as input $n = O(\poly(d/\eta))$ samples from the $d$-dimensional linear regression model with $\sigma$-additive noise (\Cref{def:linear_model}), with $\eta$-fraction of coordinate-wise corruptions according to \Cref{def:model2}, and output $\hat{\beta} \in \R^d$ that satisfies the following guarantees with probability at least $0.99$:
    \begin{itemize}[leftmargin=20pt, itemsep=0pt, topsep=5pt]
        \item Guarantee for $\cA_1$: $\|\hat{\beta} - \beta\| = O(\eta  d \sigma)$ whenever $\eta d < 0.49$.
        \item Guarantee for $\cA_2$: $\|\hat{\beta} - \beta\| = O(\eta \sqrt{d} \sqrt{ \|\beta\|^2 + \sigma^2})$, whenever $\eta < 0.24$.
        \item Guarantee for $\cA_3$: $\|\hat{\beta} - \beta\| \leq \|\beta\|$.
    \end{itemize}
\end{restatable}

Depending on $\eta$, $\sigma$, and $\|\beta\|$, a different algorithm among $\cA_1,\cA_2,\cA_3$ achieves the lowest error. \Cref{tab:eta-main,table:medium-eta,table:small-eta} in \Cref{app:intro} present this by partitioning the parameter space into regimes based on which algorithm's error is best (\Cref{sec:partition} has the explicit construction of the partition).  In \Cref{sec:unified}, we  show that $\cA_1$, $\cA_2$, and $\cA_3$ can be unified into a single algorithm that adaptively switches between the three strategies to consistently get the best error (up to constant factors).

We now present the lower bounds (the main results of this paper) in a combined statement in \Cref{thm:combined}. Since some parameter regimes in \Cref{thm:combined} overlap, we show only the strongest lower bound in \Cref{table:combined}. A more expanded version is given in \Cref{tab:eta-main,table:medium-eta,table:small-eta} which might be more convenient to look at, but due to space constraints are moved to \Cref{app:intro}.
\Cref{thm:combined} provides information-theoretic lower bounds, applying to any estimator, independent of runtime and for any sufficiently large sample size.\looseness=-1\footnote{Although there is a dependence on $n$ in the theorem statement, it stems from technicalities (our construction deletes an $\eta$-fraction in expectation rather than exactly (as in \Cref{def:model}) and this influence vanishes as $n$ grows, so the bounds remain valid even with infinite samples. Moreover, unless $n \gg d$, it is folklore that estimating $\beta$ to constant error is impossible even with clean data.}

\begin{theorem}[Lower bounds on estimation error; combined statement]\label{thm:combined}
    The following holds for every $d \in \Z_+$, $\eta \in [0,1],\sigma\in \R_+$, $b \in \R_+$, $c \in (0,1)$, and any algorithm $\cA$ that takes as input  $\eta,\sigma,b$
    as well as $n$ labeled examples $\{ (x^{(i)},y^{(i)}) \}_{i=1}^n$ with $x^{(i)} \in \R^d$ and $y^{(i)} \in \R$ and outputs a vector in $\R^d$: There exists a $\beta \in \R^d$ with $\|\beta\|=b$ such that running $\cA$ to solve the associated linear regression problem (\Cref{def:linear_model}) in the coordinate-wise corruption model with missing data (\Cref{def:model}), %
    the output $\hat{\beta}$ satisfies the following with probability at least $\frac{1}{2}-\frac{d+1}{2}e^{-\Omega(c^2 \eta n)}$ over the randomness of the clean samples:
    \begin{enumerate}[label=(\alph*)]
        \item (Small-$\beta$ regime; cf. \Cref{thm:small-beta}) If $\|\beta\| \leq \sigma$, $\cA$ has error $\Omega(\min(\|\beta\|, \eta \sqrt{d} \sigma))$.\label{it:thm-small-beta}
        \item (Large-$\eta$ regime; cf. \Cref{thm:big-eta}) If $ \tfrac{7}{\sqrt{d}} \leq \eta \leq 1$, $\cA$ has error  $\Omega(\|\beta\|)$.\label{it:thm-big-eta}
        \item (Medium-$\eta$ regime; cf. \Cref{thm:interm-eta}) If $\frac{(2+c)}{d} \leq \eta \leq \frac{7}{\sqrt{d}}$, $\cA$ has error  $\Omega( c \eta \sqrt{d}\|\beta\|)$.\label{it:thm-interm-eta}
        \item (Small-$\eta$ regime; cf. \Cref{thm:small-eta}) If $0 \leq \eta \leq \frac{1}{d}$, $\cA$ has error  $\Omega( \min ( \eta d \sigma, \eta \sqrt{d} \|\beta\| ) )$.\label{it:thm-small-eta}
    \end{enumerate}
\end{theorem}

\begin{figure}[ht]
  \centering

  \begin{minipage}{\linewidth}
    \centering
\resizebox{\textwidth}{!}{%
\begin{tabular}{>{\centering\arraybackslash}m{0.35\textwidth}
                >{\centering\arraybackslash}m{0.28\textwidth}
                >{\centering\arraybackslash}m{0.37\textwidth}
                }
\toprule
\textbf{Parameter regime} & \textbf{\shortstack{Best lower bound\\ from \Cref{thm:combined}}} & \textbf{\shortstack{Best upper bound\\from \Cref{fact:known-algorithms}}} \\
\midrule

\begin{tabular}{@{}l@{\hspace{1em}}l@{}}$ \eta \in (0, \frac{0.49}{d})$, $\;\;\;\|\beta\| \in [\sqrt{d}\sigma, + \infty)$\end{tabular}
 & \raisebox{-0.4\height}{%
  \begin{tabular}{@{}l@{\hspace{2.5em}}l@{}}
    (part \ref{it:thm-small-eta}) & $\Omega(\eta d \sigma )$
  \end{tabular}
}
&  \raisebox{-0.4\height}{%
\begin{tabular}{@{}l@{\hspace{2.7em}}l@{}}
    (Alg.~$\cA_1$) & $O({\eta d \sigma})$
  \end{tabular}}\\
\midrule

\begin{tabular}[c]{@{}l@{\hspace{1em}}l@{}}$\eta\in [\frac{0.49}{d}, \frac{7}{\sqrt{d}})$, $\|\beta\| \in [ \sigma, + \infty)$\\ $\eta \in(0, \frac{0.49}{d}),$ $\;\;\,\|\beta\| \in [\sigma , \sqrt{d}\sigma)$\end{tabular}
 & \raisebox{-0.4\height}{%
  \begin{tabular}{@{}l@{\hspace{0.3em}}l@{}}
    (part \ref{it:thm-interm-eta})* & \multirow{2}{*}{$\Omega( \eta \sqrt{d}  \|\beta\| )$} \\
    (part \ref{it:thm-small-eta}) &
  \end{tabular}
}
& 
  \begin{tabular}{@{}l@{\hspace{0.6em}}l@{}}
    \multirow{2}{*}{(Alg.~$\cA_2$)} & \multirow{2}{*}{$O(\eta \sqrt{d}  \|\beta\| ))$} 
  \end{tabular}
\\
\midrule

\begin{tabular}[c]{@{}l@{\hspace{1em}}l@{}}$\eta\in[\frac{0.49}{d}, \frac{7}{\sqrt{d}})$, $\|\beta\| \in [ \eta \sqrt{d}\sigma, \sigma)$\\ $\eta \in(0, \frac{0.49}{d}),$ $\;\;\,\|\beta\| \in [\eta \sqrt{d}\sigma , \sigma)$\end{tabular}
 & \raisebox{-0.4\height}{%
  \begin{tabular}{@{}l@{\hspace{1.8em}}l@{}}
    (part \ref{it:thm-small-beta}) & \multirow{2}{*}{$\Omega(\eta \sqrt{d} \sigma)$} \\
    (part \ref{it:thm-small-beta}) &
  \end{tabular}
}
& 
  \begin{tabular}{@{}l@{\hspace{1.2em}}l@{}}
    \multirow{2}{*}{(Alg.~$\cA_2$)} & \multirow{2}{*}{$\;\;\;O( \eta \sqrt{d} \sigma )$} 
  \end{tabular}
\\
\midrule

\begin{tabular}[c]{@{}l@{\hspace{1em}}l@{}}$\eta \in [\frac{7}{\sqrt{d}}, 1],$  $\hspace{1.3em}\|\beta\| \in[ 0, + \infty)$ \\ $\eta \in (0, \frac{7}{\sqrt{d}}),\hspace{1.1em}  \|\beta\| \in [0, \eta \sqrt{d}\sigma)$ \end{tabular}
& \raisebox{-0.4\height}{%
  \begin{tabular}{@{}l@{\hspace{2.5em}}l@{}}
    (part \ref{it:thm-big-eta}) & \multirow{2}{*}{$\Omega(\|\beta\|)$} \\
    (part \ref{it:thm-small-beta}) &
  \end{tabular}
}
& \begin{tabular}{@{}l@{\hspace{2.7em}}l@{}}
    \multirow{2}{*}{(Alg.~$\cA_3$)} & \multirow{2}{*}{$O( \|\beta\|)$} 
  \end{tabular}
\\
\bottomrule
\end{tabular}%
}
\caption{Comparison of error bounds across the different possible parameter regimes. *\Cref{thm:combined}\ref{it:thm-interm-eta} requires $ \eta \in [\frac{2+c}{d},\tfrac{7}{\sqrt{d}}]$ where $c > 0$ can be any arbitrarily small absolute constant.\looseness=-1}
\label{table:combined}
\end{minipage}
\end{figure}

A few comments are in order regarding the theorem statement. First, the failure probability in the theorem statement is close to a half (instead of one) because the lower bound is obtained from a reduction to a hypothesis testing problem that asks for distinguishing between two parameter vectors $\beta, \beta'$ (see \Cref{sec:warm-up} for more details), thus a trivial algorithm that outputs each of $\beta, \beta',$ with probability $1/2$ succeeds with the same probability. Second, parts \ref{it:thm-big-eta} and \ref{it:thm-interm-eta} hold under arbitrary additive noise, not necessarily Gaussian. This is because the conclusion as stated in the theorem holds in the noiseless case ($\sigma=0$), thus by a straightforward reduction this implies hardness under any noise distribution: given an estimator for the noisy model, one could add the same noise to noiseless hard instance and then apply the estimator. Finally, in terms of the covariate distribution, \Cref{thm:combined} shows that the problem is hard even for the most restricted and basic setting of Gaussian covariates; the hardness trivially extends to any family of distributions that contains Gaussians.

We now discuss the conceptual takeaways and how \Cref{table:combined} addresses our main research questions.

\paragraph{Missing vs.~Replaced Data} The upper bounds hold under the stronger model of data replacement (\Cref{def:model2}), while the lower bounds assume the weaker model of missing data (\Cref{def:model}). Since the bounds match (up to constants) for essentially the entire range of $\eta$, $\sigma$, and $\|\beta\|$, both  models are equally hard in terms of optimal estimation error.
While this may be unsurprising for the robust statistics framework where samples are either entirely corrupted or uncorrupted, our more delicate lower bound constructions show that the fact remains true even under coordinate-wise erasure/corruption.\looseness=-1

\paragraph{Corruption Threshold for Non Trivial Estimation} In the classical robust statistics model with $\eps$-fraction of (entirely) corrupted samples, estimation becomes impossible at $\eps = 1/2$, since half the samples could come from two different models and there is no way to tell which are the inliers. \Cref{thm:combined}\ref{it:thm-big-eta} shows that the analogous transition point for \Cref{def:model,def:model2} is $\eta = \Theta(1/\sqrt{d})$. This is interesting, as there is no equally obvious explanation for this threshold in this model.

\paragraph{$\|\beta\|$-Dependence on Error is Necessary} Ignoring the regime $\eta \geq 1/\sqrt{d}$, where meaningful estimation is impossible, our results show that $\|\beta\|$-dependence in the error is unavoidable in general %
---in stark contrast to the classical robust model (with each sample either fully clean or corrupted), where optimal error is independent of $\|\beta\|$. Here is the intuition: in that setting, a variant of $\cA_2$ gives error $O(\eps(\|\beta\| + \sigma))$ (with $\eps$ corruption rate). But by shifting labels via $y \gets y - \hat{\beta}^\top X$, one reduces the regression target to $\beta - \hat{\beta}$, which has smaller norm, and can recursively apply the algorithm to eliminate the $\|\beta\|$-dependence. This technique fails under \Cref{def:model2}: even one corrupted coordinate in $X$ corrupts $\hat{\beta}^\top X$, and the transformation $y \gets y - \hat{\beta}^\top X$ can amplify the label corruption rate from $\eta$ to $\eta d$ in a single step. Our results show that no workaround exists.

\paragraph{Role of Label Corruptions} As can be inferred from the proofs, all of our lower bounds hold even if $y$ is never masked.
Thus, the difficulty in estimating $\beta$ comes only from the missing $X$-coordinates.

\paragraph{Additional Discussion}  \Cref{thm:combined}\ref{it:thm-big-eta} may extend to other activation functions. Since \Cref{thm:combined}\ref{it:thm-big-eta} and \ref{it:thm-interm-eta} hold for $\sigma=0$, they imply that for any estimator, the mean squared prediction error can be bounded below for a broad class of GLMs, such as a class of GLMs with non-decreasing Lipschitz activations, including ReLUs in particular.
See \Cref{app:discussion} for the relevant discussion.\looseness=-1

\subsection{Overview of Techniques}\label{sec:overview}
\vspace{-3pt}

The general approach for our lower bounds is the following. We first reduce estimation to a hypothesis testing problem. If no algorithm can distinguish between the regressor being $\beta$ vs.~$\beta'$, then no estimator can achieve error less than $\|\beta - \beta'\|/2$. To show testing hardness, it suffices to construct a coupling between the distributions of the labeled examples in \Cref{def:linear_model} with regressors $\beta$ and $\beta'$ and with small coordinate-wise disagreements, that is, a pair of jointly distributed labeled examples $(X,y), (X',y')$ such that $\Pr[X_i \neq X_i'] \leq \eta$ for all $i\in [d]$ and $\Pr[y \neq y'] \leq \eta$. This is because an adversary can then delete differing coordinates, making the datasets indistinguishable (cf.~\Cref{fact:coupling}). This leaves us with the (highly) non-trivial task of constructing that coupling, outlined below.\looseness=-1

\paragraph{Basic construction}  The high-level idea of the coupling construction boils down to two main steps: (1)  induce the same $y$-marginal distribution for some sufficiently distinct $\beta$ and $\beta'$, and (2) for each value of $y$, construct a masking scheme to make $X \, | \, y$ and $X' \, | \, y$ indistinguishable.
Recall that, since $X$ (and $X'$) and the linear regression noise are all Gaussian, the conditional distributions $X \, | \, y$ and $X' \, | \, y'$ are also Gaussian; denote them by $\cN(\mu, \Sigma)$ and $\cN(\mu',\Sigma')$.  
In all considered regimes of $\eta$, our choices of $\beta$ and $\beta'$ are such that the covariances $\Sigma$ and $\Sigma'$ are the same.
Then, to obtain a masking scheme that makes the two conditional distributions the same \emph{while fully-leveraging the coordinate-wise masking capabilities of the adversary}, in \Cref{lem:hybrid}, we construct a \emph{sequence} of Gaussians $N^{(0)}, N^{(1)}, N^{(2)}, \ldots, N^{(d)}$ with the same covariance matrix, which  interpolates between $X \, | \, y$ and $X' \, | \, y$. Namely,  $N^{(0)} = X \, | \, y$ and $N^{(d)} = X' \, | \, y$, and each consecutive pair $N^{(i)}, N^{(i+1)},$ $i \in \{0, \dots, d-1\},$ has means that differ over a single coordinate. 
We then construct a masking scheme for the pair $N^{(i)}$ and $N^{(i+1)}$ using a \emph{maximal coupling}\footnote{There is always a coupling (called \emph{maximal}) between $P,Q$ with $\Pr[X \neq X'] = \DTV(P,Q)$ (\Cref{fact:maximal}).} that masks only the $i^\text{th}$ coordinate, with probability $\DTV(N^{(i)}, N^{(i+1)})$ (cf.~\Cref{cl:one_step}), and that makes the pair of Gaussians indistinguishable.

\paragraph{Case $\|\beta\| \leq \sigma$}  The above basic construction is sufficient for $\|\beta\| \leq \sigma$ (\Cref{thm:combined}\ref{it:thm-small-beta}).
In this case, we choose $\beta = (r, (b/\sqrt{d}) \1_{d-1})$ and $\beta' = -\beta$, where $b \leq \eta \sqrt{d}\sigma$,  $\1_\ell$ denotes the dimension-$\ell$ vector of all ones, and $r$ is chosen to set the norms of $\beta, \beta'$ to the target value.

\paragraph{Handling the case $\sigma/\|\beta\| \to 0$ and $\eta \gg 1/\sqrt{d}$}\footnote{Henceforth, the notation $a \gg b$ indicates that where exists some sufficiently large constant $C > 0$ such that $a \geq C b$, and $\ll$ is defined similarly. For the discussion in this section, $C \in [1, 4].$} The basic construction described above is only meaningful (both in terms of the error and the proof itself) when $0 < \|\beta \| \leq \sigma$. As $\sigma/\|\beta\| \to 0$, the conditional distributions $X\,|\,y$ and $X'\,|\,y$  approach Gaussians with singular covariances, making the TV-distances between the conditionals equal to one. Thus, the distributions cannot be made indistinguishable by deleting only an $\eta < 1$ fraction of coordinates. 

Our strategy in this case is to effectively reduce the data dimension by one and treat one of the coordinates as noise.
For this discussion, fix $\|\beta\| = \sqrt{d}$, let $\beta=\1_d$ and $\beta' = -\beta$; the same ideas transfer to other values of $\|\beta\|$.
The core task again is to couple $X \, |\, y$ with $X' \,|\, y'$.
Now, given the particular form of $\beta$ and $\beta'$ and that $\beta = -\beta'$, the core task can be re-phrased as coupling 
$X \sim \cN(0,I)$ conditioned on $\sum_i X_i = t$ with $X'$ conditioned on $\sum_i X_i' = t'$ where $t = -t'$.
Since there is no noise ($\sigma = 0$), the key idea is to treat the last coordinate as ``noise'': we couple the first $d-1$ coordinates using the basic coupling method, which gives $\Pr[X_i \neq X_i'] \leq |t - t'|/d$ for $i \in [d-1]$.
We then set the final coordinate to ensure the total sums are $t$ and $t'$.
This makes the last coordinate differ with probability 1 between the two scenarios, but randomizing the coordinate order leads to coupled $X, X'$ where each coordinate differs with probability $(|t - t'| + 1)/d$.
This idea and construction is formally stated in~\Cref{lem:impr-coupling} (which gives a construction for general $t$ and $t'$ instead of only $t = -t'$).
Finally, since in expectation $|t - t'| = 2|t| = 2|y|$ is $O(\sqrt{d})$, this yields a bound of $\Omega(\|\beta\|)$ as long as $\eta \gg 1/\sqrt{d}$, leading to \Cref{thm:combined}\ref{it:thm-big-eta}.

\paragraph{Handling $1/d \ll \eta \ll 1/\sqrt{d}$} 
When $1/d \ll \eta \ll 1/\sqrt{d}$, the approach described so far fails because  $\E[|t - t'|]$ is too large, making the probability of coordinate-disagreement larger than the assumed $\eta$ upper bound. 
Instead, we construct a new hard instance with $\beta$ and $\beta'$ proportional to $ (\eps \1_{d/2}, \1_{d/2})$ and $(-\eps \1_{d/2}, \1_{d/2})$ respectively. The coupling proceeds as follows: sample the labels $y = y'$ from the distribution $\cN(0,(1+\eps^2)/d))$, which is the distribution of the labels.
Now consider coupling $t := \sum_{i=1}^{d/2} X_i$ and $t' := \sum_{i=1}^{d/2} X_i'$, which are the sums of the first half of the coordinates under each hypothesis.
From calculations, the marginals $t$ and $t'$ are both univariate Gaussians with the same variance, with difference in mean equal to $\Theta(\eps y)$.
We thus sample $t$ and $t'$ in a coupled way where $t - t' = \Theta(\eps y)$ always (this is possible as the distribution of $t$ is the same as the distribution of $t'$ shifted by $\Theta(\eps y)$).
Now it remains to couple $X \,|\, y$ and $X'\,|\,y'$, which we do by \emph{first} coupling the first half of the coordinates using the coupling between $t$ and $t'$, then coupling the second half of the coordinates based on the values of $y,t$ and $y',t'$.
More concretely, we use~\Cref{lem:impr-coupling} to couple $(X_1,\ldots,X_{d/2}) \,|\, t$ and $(X'_1,\ldots,X'_{d/2}) \,|\, t'$.
Then, we use~\Cref{lem:impr-coupling} again to couple $(X_{d/2+1},\ldots,X_d)$ and $(X'_{d/2+1},\ldots,X'_d)$ so their sums are $y - \eps t$ and $y + \eps t'$, respectively. Going through the calculations, this strategy yields an expected number of disagreeing coordinates bounded by $\Theta(1 + \eps\sqrt{d})$, which is at most $\eta d$ as long as $\eps \ll \eta\sqrt{d}$ (recall that we assume $\eta \gg 1/d$ in this paragraph).
This translates to an estimation error lower bound of $\Omega(\eta \sqrt{d}  \|\beta\|)$, proving \Cref{thm:combined}\ref{it:thm-interm-eta}.\looseness=-1

\paragraph{The case of small $\eta$}  
The remaining regime is $\eta \ll 1/d$ (\Cref{thm:combined}\ref{it:thm-small-eta}).
The (already complicated) construction in the last paragraph does not handle that case, given that the expected number of disagreeing coordinates there is $\Omega(1) \ge 1/d \gg \eta$.
For this final regime, we use a similar hard hypothesis test instance, showing the impossibility of distinguishing between $\beta = (B\1_{d/2},E \1_{d/2})/\sqrt{d/2}$ and $\beta' = (B\1_{d/2},-E\1_{d/2})/\sqrt{d/2}$ for  appropriately chosen parameters $B$ and $E$.
The main technical innovation here is a deconstruction of label values into the contribution from the first and second half of the coordinates and a more refined way to sample them in the coupling construction.
This is sketched in \Cref{sec:thm-proofs} and detailed in \Cref{app:thm-proofs}.

\subsection{Preliminaries}\label{sec:prelims}

We include only the essential preliminaries here; the full version appears in \Cref{app:prelims}.
For $(X,Y) \sim P$, we write $X$ for the marginal of $X$ and $X|(Y=y)$ for the conditional. We denote expectations as $\E_{X \sim D}[X]$. For $X,Y$ with pdfs $P_X$ and $P_Y$, we write $\dtv(X,Y)$ or $\dtv(P_X,P_Y)$ for the total variation distance, defined as $\tfrac{1}{2}\int |P_X(u) - P_Y(u) |\d u$. We write $\DKL(X \;||\; Y)$ or $\DKL(P_X \;||\; P_Y)$ for the Kullback–Leibler divergence, defined as $\DKL(P_X \;||\; P_Y) = \int_x P_X(z) \log (P_X(z)/P_Y(z))\d z$.
We use $\R_+$ for non-negative reals, $\Z_+$ for non-negative integers, and $[n] = {1,\dots,n}$. For a vector $x$, $\|x\|$ denotes its Euclidean norm. Let $I_d$ be the $d \times d$ identity matrix, and $\1_d$ the all-ones vector in $\R^d$. We use $\top$ to denote matrix transposes.
We write $a \lesssim b$ (or $O(b)$) to mean $a \leq Cb$ for an absolute constant $C > 0$, independent of the parameters involved (similarly for $\Omega(\cdot)$). We write $\tilde O$ and $\tilde \Omega$ to hide polylog factors.\looseness=-1

\begin{restatable}[see, e.g., Section 8.1.3.~in \cite{PP2008}]{fact}{FACTCONDITIONAL}\label{fact:conditional}
    If $X \sim \cN(\mu, I)$ is a multivariate Gaussian vector in $\R^d$, $\xi \sim \cN(0,\sigma^2)$ is a univariate Gaussian, and $u \in \R^d$ is a fixed (deterministic) vector, then the distribution of $X$ conditioned on $u^\top X + \xi = r$ is
        $\cN\left( r u/(\|u\|^2+ \sigma^2), \; I - u u^\top /(\|u\|^2 + \sigma^2)  \right)$.
\end{restatable}

If $P$ and $Q$ are distributions over $\cX$, a coupling $\Pi$ of $P$ and $Q$ is any distribution over $\cX \times \cX$ with marginals $P$ and $Q$ such that under $(X,Y) \sim \Pi$, the marginals are $X \sim P$ and $Y \sim Q$.

\begin{restatable}[Maximal coupling (see, e.g., \cite{Roc2024})]{fact}{FACTMAXIMAL}\label{fact:maximal}
    Let $P$ and $Q$ be distributions on some domain $\cX$.
    There exists a coupling $\Pi$ between $P$ and $Q$ such that $\Pr_{(X,Y) \sim \Pi}[X \neq Y] = \dtv\left( P ,Q \right)$.
\end{restatable}

\begin{restatable}[Total variation between univariate Gaussians (see, e.g., \cite{PP2008}]{fact}{FACTTVGAUSSIANS}\label{fact:TVDGaussians}
    If $D_1 = \cN(\mu_1, \sigma^2)$ and $D_2 = \cN(\mu_1, \sigma^2)$ then $\dtv\left( D_1, D_2 \right) \leq (1/\sqrt{2}){|\mu_1-\mu_2|}/{\sigma}$.
\end{restatable}

\section{Warm-up: Coupling via Hybrid Argument and Proof of \Cref{thm:combined}\ref{it:thm-small-beta}}\label{sec:warm-up}

We show that our goal reduces to finding a coupling with bounded coordinate-wise disagreements. We also present a method (the ``hybrid argument'') for building such couplings. This alone suffices for \Cref{thm:combined}\ref{it:thm-small-beta}, but stronger results require a modification that we  give in \Cref{sec:improved-coupling}.\looseness=-1

\paragraph{Estimation Hardness via Hypothesis Testing}\label{sec:reduction}
To show linear regression hardness with error $<\eps$, it suffices to show 
an indistinguishability result
for vectors $\beta^{(0)},\beta^{(1)}$ with $\|\beta^{(0)} - \beta^{(1)}\| > 2\eps$:\looseness=-1
\begin{itemize} 
\item (Null Hypothesis) The data follows the model in \Cref{def:linear_model,def:model} with $\beta = \beta^{(0)}$.
\item (Alternative Hypothesis) The data follows the model in \Cref{def:linear_model,def:model} with $\beta = \beta^{(1)}$.
\end{itemize}
This is sufficient because if an estimator $\hat{\beta}$ exists with error less than $\eps$, the test that rejects the null hypothesis when $\|\hat{\beta} - \beta^{(0)}\| > \|\hat{\beta} - \beta^{(1)}\|$ will succeed, as shown by a simple triangle inequality.

\paragraph{Couplings with Small Coordinate-Wise Disagreements}\label{sec:couplings-to-hardness}
We now argue that constructing a suitable coupling suffices for proving hardness of the testing problem. One can think of data generation as follows: a pair of two samples $((X,y), (X',y'))$ is drawn from a coupling $\Pi$ between the distributions of \Cref{def:linear_model} with $\beta = \beta^{(0)}$ and $\beta = \beta^{(1)}$. Depending on the hypothesis, either $(X,y)$ or $(X',y')$ is added to the dataset, after which the adversary may modify it (\Cref{def:model}). If $\E[\1(X_i \neq X_i')] \leq \eta$ for all $i \in [d]$, then by deleting disagreeing coordinates, the adversary can make the dataset look identical under both hypotheses, making testing impossible. \Cref{fact:coupling} formalizes this. This lemma actually uses a slightly relaxed coordinate-wise disagreement condition: $\E[ \sum_i \1(X_i \neq X_i)] \leq  \eta d$. At first glance, this may appear to correspond to a slightly different definition of the contamination model, in which the adversary is allowed to corrupt $\eta d$ coordinates \emph{per sample} in expectation. However, \Cref{def:model} instead imposes a budget of $\eta n$ corruptions \emph{per coordinate}. We can overcome this issue by assuming that the distribution of $(X,y)$ is invariant to permutations of the coordinates of $X$, because this would imply that after a random permutation each coordinate is deleted with equal probability (at most $\eta$). Moreover, because the adversary's budget in \Cref{def:model} is deterministic (rather than in expectation), there is a small probability—dictated by the Chernoff–Hoeffding bound—that our randomized adversary does not satisfy \Cref{def:model}. In this case, we may simply assume that the testing problem is solvable.
This completes the proof sketch of \Cref{fact:coupling}. Finally, we note that we will actually require the permutation invariance condition in the statement of \Cref{fact:coupling} to hold separately over the first and second halves of the coordinates, rather than over the entire vector. This modification ensures that the lemma remains applicable in the more involved construction of \Cref{thm:combined}\ref{it:thm-small-eta}, which treats the two halves of the coordinates independently. The formal proof is given in \Cref{app:warm-up-couplings}.

\begin{restatable}{lemma}{FACTCOUPLING}\label{fact:coupling}
Let $D,D'$ be distributions over labeled examples $(X, y)$ with $X \in \R^d$, $y \in \R$. Assume that for any two permutations $\pi_1, \pi_2 \in S_{d/2}$, the distribution of $(X,y) \sim D$ is the same as that of $(\pi_1(X_1), \pi_2(X_2), y)$, where $X_1$ and $X_2$ denote the first and last $d/2$ coordinates of $X$, respectively. Assume that the same property also holds for $D'$. Consider the hypothesis testing problem where the null hypothesis is that data are drawn from $D$ under the corruption model of \Cref{def:model}, and the alternative hypothesis is that data are drawn from $D'$.
If there exists a coupling $\Pi$ between $D,D'$ such that for some $c \in (0,1)$: $\Pr_{((X,y),(X',y')) \sim \Pi}\left[ y \neq y' \right] \leq \frac{\eta}{2} (1 - c)$ and \looseness=-1
\begin{align} \label{eq:coupling-assumption}
   \textstyle \E\limits_{((X,y),(X',y')) \sim \Pi} \left[  \sum_{i=1}^d \1\left( X_i \neq X_i'  \right) \right] \leq \frac{\eta}{2} d (1 - c),
\end{align}
then no test can distinguish $D$ from $D'$ with probability greater than $\frac{1}{2}+\frac{(d {+} 1)}{2} e^{-\Omega(c^2 \eta n)}$.
\end{restatable}

\paragraph{Hybrid Argument for Coupling Construction}\label{sec:hybrid}
We discuss an argument for constructing couplings satisfying \Cref{eq:coupling-assumption}. Instead of coupling the distributions of \Cref{def:linear_model}, we will first show, in \Cref{lem:hybrid} below, how to construct such couplings between two $d$-dimensional Gaussians $D$ and $D'$ with same covariance and different means. As we will see in the next subsection, this directly leads to a coupling between the distributions of the linear regression model. \looseness=-1

\begin{restatable}[Hybrid argument for constructing couplings]{lemma}{LEMMAHYBRID}\label{lem:hybrid}
    Let $D = \cN((\mu_1,\ldots,\mu_d), \Sigma)$ and $D' = \cN((\mu'_1,\ldots,\mu'_d), \Sigma)$. 
    There exists a coupling $\Pi$ between $D$ and $D'$ such that
    \begin{align*}
         \E_{(X,X') \sim \Pi}\left[ \sum_{i=1}^d \1(X_i \neq X_i')  \right] {=} \sum_{i=0}^{d-1} \DTV\left( Q_i, Q_{i+1} \right),
    \end{align*}
    where $Q_i = \cN((\mu'_1\mu'_2,\ldots,\mu'_i,\mu_{i+1},\ldots,\mu_d), \Sigma)$ is the Gaussian whose mean has the same values as $\mu'$ in the first $i$ coordinates and the same as $\mu$ elsewhere ($Q_0 = D$ and $Q_d = D'$).
\end{restatable}

The idea is the following: if the two means differ only in the $i$-th coordinate, we can adapt the standard maximal coupling (\Cref{fact:maximal}) so that any disagreement comes exclusively from the $i$-th coordinate.\looseness=-1

\begin{restatable}{claim}{CLAIMONESTEP}\label{cl:one_step}
        Consider the two $d$-dimensional Gaussians $Q = \cN((\mu_1,\mu_2,\ldots,\mu_d), \Sigma)$ and $Q' = \cN((\mu_1,\ldots,\mu_{i-1}, \mu_{i}',\mu_{i+1},\ldots,\mu_d), \Sigma)$. There is a coupling $\Pi$ between the distributions $Q$ and $Q'$ such that 
        $\Pr_{(X,X') \sim \Pi}[X_i \neq X_i'] = \DTV(Q,Q')$ and $\Pr_{(X,X') \sim \Pi}[X_j \neq X_j']=0$ for all $j \neq i$.
\end{restatable}
\noindent\Cref{cl:one_step} holds because the marginals over $[d] \setminus \{i\}$ are identical under $Q$ and $Q'$, allowing us to match these coordinates, and apply \Cref{fact:maximal} to the distribution of the $i$-th coordinate conditioned on the rest.\looseness=-1

To prove \Cref{lem:hybrid} for means differing in more coordinates, we can consider the distributions $Q_i$ from the statement and couple each consecutive pair using \Cref{cl:one_step}.
Chaining together the couplings completes the proof of \Cref{lem:hybrid}. The formal proofs are in \Cref{app:warm-up-couplings}.

\paragraph{Proof Sketch of \Cref{thm:combined}\ref{it:thm-small-beta}}
  We present a proof sketch with routine calculations omitted here; the full proof is in \Cref{app:proof-of-small-beta}. We first focus on the regime $\|\beta\| \leq \eta \sqrt{d} \sigma/(2\sqrt{2})$. The remaining regime can be handled by an argument that essentially reduces to the $\|\beta\| \leq \eta \sqrt{d} \sigma/(2\sqrt{2})$ regime and will be explained at the end.
For any $b \in [0,\eta \sqrt{d}\sigma/(2\sqrt{2})]$, we consider the hypothesis testing problem of distinguishing between the regression vectors $\beta^{(0)} =  (b/\sqrt{d},\ldots,b/\sqrt{d})$ and $\beta^{(1)} = - \beta^{(0)}$ from $n$ samples from the models in \Cref{def:linear_model,def:model}. 
Note that $\|\beta^{(0)}\| = \| \beta^{(1)} \| = b$ and $\|\beta^{(0)} - \beta^{(1)}\| = \sqrt{2} b$.
By the reduction from estimation to hypothesis testing, showing that no algorithm can solve this testing problem implies that no estimator has Euclidean error smaller than $b/\sqrt{2}$.\looseness=-1

To use \Cref{fact:coupling}, we need a coupling $(X,y),(X',y')$ where $(X,y)$ follows \Cref{def:linear_model} with $\beta = \beta^{(0)}$ and $(X',y')$ follows the same model with $\beta = \beta^{(1)}$, satisfying $\E[\sum_{i=1}^d \1(X_i = X_i')] \leq \eta d$ and $\Pr[y \neq y'] \leq \eta$. We construct this coupling by: (i) drawing $t$ from $\cN(0,\|\beta\|^2 + \sigma^2)$, and setting $y = y' = t$ (the label distribution is the same for both hypotheses), and (ii) drawing $X$ and $X'$ from an appropriate coupling of the conditional distributions given $y = t$. We thus focus on how to do step (ii) with small coordinate-wise disagreements. 

\begin{claim}\label{cl:helper}
    For $t \in \R$, $\sigma > 0$ and $b \in [0,\eta \sqrt{d} \sigma/(2\sqrt{2})]$, let $D$ be the distribution of $X |(y=t)$ under $X \sim \cN(0,I_d)$ and $y = (b/\sqrt{d},\ldots,b/\sqrt{d})^\top X + \xi$ for $\xi \sim \cN(0,\sigma^2)$. Let $D'$ be defined similarly but using the conditioning $y=-t$. There exists a coupling $\Pi$ between $D$ and $D'$ such that $\E_{(X,X') \sim \Pi}\left[ \sum_{i=1}^d \1( X_i \neq X_i') \right] \leq   \sqrt{2d} |t|  b/(\sigma^2 + b^2)$.
\end{claim}

\begin{proof}
Let $\beta^{(0)} = (b/\sqrt{d},\ldots,b/\sqrt{d})$ and $\beta^{(1)} = -\beta^{(0)}$ as before.
By \Cref{fact:conditional}, $D^{(i)}$ for $i=0,1$ are the Gaussians $\cN(\mu^{(i)},\Sigma^{(i)})$ where 
    \begin{align*}%
        \mu^{(i)} := \frac{t}{\sigma^2+ b^2} \beta^{(i)},\quad  \Sigma^{(1)} {=} \Sigma^{(2)} {=} \Sigma := I {-} \frac{\beta^{(1)}(\beta^{(1)})^\top}{\sigma^2 + b^2}. 
    \end{align*}
\noindent 
We now use \Cref{lem:hybrid} to couple $\cN(\mu^{(0)},\Sigma),\cN(\mu^{(1)},\Sigma)$. For each $i \in [d]$ we need to upper bound the TV-distance between the two Gaussians $\cN(m^{(i)}, \Sigma), \cN(m^{(i+1)} , \Sigma)$, where $m^{(i)} = (\mu^{(1)}_1, \ldots, \mu^{(1)}_{i-1}, \mu^{(0)}_{i}, \mu^{(0)}_{i+1}, \ldots, \mu^{(0)}_{d})$ and $m^{(i+1)} = (\mu^{(1)}_1, \ldots, \mu^{(1)}_{i-1}, \mu^{(1)}_{i}, \mu^{(0)}_{i+1}, \ldots, \mu^{(0)}_{d})$.
First, by Pinsker's inequality it suffices to bound $\tfrac{1}{\sqrt{2}}\DKL(\cN(m^{(i)} , \Sigma)  \,\|\, \cN(m^{(i+1)} , \Sigma))^{1/2}$.
The KL-divergence can then be bounded by
(see \Cref{eq:firsteq}-(\ref{eq:lasteq}) in \Cref{app:proof-of-small-beta} for the missing steps):
\begin{align}
    \textstyle &\sqrt{\DKL(\cN(m^{(i)} , \Sigma)  \,\|\, \cN(m^{(i+1)} , \Sigma)) }  
    \leq \sqrt{\tfrac{1}{2}(m^{(i+1)} - m^{(i)})^\top \Sigma^{-1} (m^{(i+1)} - m^{(i)})}\label{eq:boundDTV-via-pinsker1}\\
    &\leq \sqrt{\tfrac{1}{2}(m^{(i+1)} - m^{(i)})^\top (I + \tfrac{\beta^{(1)}{\beta^{(1)}}^\top}{\sigma}) (m^{(i+1)} - m^{(i)})}
    \leq \frac{\sqrt{2} |t|}{\sqrt{d}}  \frac{b}{\sigma^2 + b^2} . \notag
\end{align}
Adding together all the TV-distance terms for $i=1,\ldots,d$ concludes the proof of \Cref{cl:helper}.
\end{proof}
\vspace{-10pt}
By turning this into a coupling for the full labeled examples as described in the beginning of this proof sketch, the expected sum of disagreements is bounded by simply taking expectation with respect to $t \sim \cN(0,\sigma^2 + b^2)$ above. By using $\E[|t|] \leq \sqrt{\sigma^2 + b^2}$ and plugging in $b \leq \eta \sqrt{d} \sigma/(2\sqrt{2})$, the expected sum coordinate-wise disagreement becomes at most $\eta d/2$. Applying \Cref{fact:coupling} completes the proof of the lower bound $\Omega(\|\beta\|)$ in the regime $\|\beta\| \leq \eta \sqrt{d} \sigma/(2\sqrt{2})$.

For the remaining regime $\|\beta\| > \eta \sqrt{d} \sigma /(2\sqrt{2})$, consider instead the  testing problem in $\R^{d}$ between $\beta = (r, b/\sqrt{d}, \ldots, b/\sqrt{d})$ and $\beta' = (r, -b/\sqrt{d}, \ldots, -b/\sqrt{d})$, where $b \leq \eta \sqrt{d}\sigma/(2\sqrt{2})$ is as before and $r$ is a tunable parameter allowing $\|\beta\|$ to get any desired value larger than $\eta \sqrt{d} \sigma/(2\sqrt{2})$. The labels can be written as $y = r X_1 + y_0$ and $y' = r X_1' + y_0'$, where $y_0 = (b/\sqrt{d}, \ldots, b/\sqrt{d})^\top X_{2:d} + \xi$ and $y_0' = (-b/\sqrt{d}, \ldots, -b/\sqrt{d})^\top X_{2:d}' + \xi'$, with $X_{2:d}$ denoting the last $d$ coordinates of $X$. From earlier, there exists a coupling between $(X_{2:d}, y_0)$ and $(X_{2:d}', y_0')$ with expected disagreements at most $\eta d/2$ per sample which extends trivially to the full $(X, y)$ and $(X', y')$ by adding shared Gaussian noise $\mathcal{N}(0, r^2)$ to the labels. Applying \Cref{fact:coupling} as before completes the proof. 
\qed

\section{Improved Core Coupling Construction}\label{sec:improved-coupling}

The previous strategy has a key limitation that the result holds only for $\|\beta\| \leq \sigma$. When $\sigma = 0$, the matrix $\Sigma$ in \Cref{eq:boundDTV-via-pinsker1} becomes singular, causing the proof to break down and yield no meaningful bound. For some regimes in \Cref{thm:combined}, we want to overcome this, for example, to establish the lower bound $\Omega(\|\beta\|)$ whenever $\eta \geq 7/\sqrt{d}$, regardless of the value of $\sigma$. The key observation is that, because $\beta$ in \Cref{cl:helper} points in the all-ones direction, the zero-eigenvalue issue disappears when we restrict our attention to the first $d-1$ coordinates. We can thus apply the technique to those coordinates and adjust the final one to ensure the total sum is correct. This idea is formalized in \Cref{lem:impr-coupling}, which provides an improved version of \Cref{cl:helper} for the case $b = \sqrt{d}$ and $\sigma = 0$ (the result extends trivially to any scaling of the vector and any $\sigma > 0$). We will later use \Cref{lem:impr-coupling} to prove our remaining main results.\looseness=-1

\begin{restatable}[Improved core coupling]{lemma}{IMPRCOUPLING} \label{lem:impr-coupling}
        For any $d \in \Z_+$,  $t,t'  \in  \R$, the following hold. 
    If $D$ denotes  the distribution of $(X_1,\ldots, X_d) \sim \cN(0,I_d)$ conditioned on $\sum_{i=1}^d X_i = t$ and $D'$ the distribution of $(X'_1,\ldots, X'_d)\sim \cN(0,I_d)$ conditioned on $\sum_{i=1}^d X'_i  =  t'$, then there exists a coupling $\Pi$ between $D,D'$ such that  $\E_{(X,X')) \sim \Pi } \left[ \sum_{i=1}^d \1\left( X_i {\neq} X'_i \right) \right] \leq  1 +  |t - t'| $.\looseness=-1
\end{restatable}
\begin{proof}
    Consider generating the pair $(X,X')$ as follows:
    \begin{itemize} 
        \item Let $\Pi'$ be the coupling from \Cref{lem:hybrid} between the first $(d-1)$ coordinates of $D$ and $D'$. Sample $(X_1,\ldots,X_{d-1}),(X_1',\ldots,X_{d-1}')$ from $\Pi'$.
        \item Set the last coordinate as $X_{d} = t - \sum_{i=1}^{d-1}X_i$ and $X'_{d} = t' - \sum_{i=1}^{d-1}X'_i$.
    \end{itemize}
    By construction, $X \sim D$ and $X' \sim D'$, so the above defines a valid coupling between $D$ and $D'$. We denote this coupling by $\Pi$, and claim that it satisfies the lemma's guarantee. By \Cref{lem:hybrid} and the trivial bound $\1(X_d \neq X'_d)\leq 1$ on the $d$-th coordinate, we have
    \begin{equation}\label{eq:coupling_dis_bound}
         \E_{(X,X')\sim \Pi}\left[ \sum_{i=1}^d \1(X_i {\neq} X_i')  \right] 
        {\leq }1 {+} \sum_{i=0}^{d-2} \DTV\left( Q_i, Q_{i+1} \right) ,
    \end{equation}
    where the $Q_i$ are as defined in \Cref{lem:hybrid}. It remains to determine their exact form and bound their total variation distances. Applying \Cref{fact:conditional} with $u {=} \1_{d-1}$ (the all-ones vector in $\R^{d-1}$) and $\sigma {=} 0$, we find that the first $d-1$ coordinates of $D$ and $D'$ follow $\cN(\mu, \Sigma)$ and $\cN(\mu', \Sigma)$, respectively, where $\mu {=} \1_{d-1} t/d$, $\mu' {=} \1_{d-1} t'/d$, and $\Sigma = I_{d-1} {-} \1_{d-1} \1_{d-1}^\top / d$. Thus, each $Q_i$ is a Gaussian with covariance $\Sigma$ and mean $\mu^{(i)}$, where $\mu^{(i)}$ equals $t'/d$ on the first $i$ coordinates and $t/d$ on the remaining ones. By Pinsker's inequality, each term $\dtv(Q_i, Q_{i+1})$ in \Cref{lem:hybrid} is \looseness=-1
    \begin{align}
        &\dtv(Q_i, Q_{i+1}) \leq \sqrt{\tfrac{1}{2}\DKL(Q_i  \,\|\, Q_{i+1})} 
        \leq   \sqrt{\tfrac{1}{4}(\mu^{(i+1)} - \mu^{(i)})^\top \Sigma^{-1}(\mu^{(i+1)} - \mu^{(i)})} \tag{\Cref{fact:KL-Gaussians}}\\
        &= \sqrt{\tfrac{1}{4}(\mu^{(i+1)} - \mu^{(i)})^\top \left( I +  \1_{d-1} \1_{d-1}^\top\right)(\mu^{(i+1)} - \mu^{(i)})} \tag{\Cref{cl:sherman-morrison}}\\
        &\leq \sqrt{\tfrac{1}{4} \left( \|\mu^{(i+1)} - \mu^{(i)}\|^2 + |(\mu^{(i+1)} - \mu^{(i)})^\top \1_{d-1}|^2  \right)   } \notag\\
        &= \sqrt{\tfrac{1}{4} \left( |t-t'|^2/d^2 + |t-t'|^2/d^2  \right)   } 
        = \frac{1}{\sqrt{2}}\frac{|t-t'|}{d}.
    \end{align}
    Thus, plugging this to \Cref{eq:coupling_dis_bound} we obtain $\E_{(X,X') \sim \Pi} [ \sum_{i=1}^d \1(X_i \neq X_i')  ] \leq 1 + |t-t'|$.
\end{proof}

\section{Proofs of parts \ref{it:thm-big-eta}-\ref{it:thm-small-eta} of \Cref{thm:combined}}\label{sec:thm-proofs}

The proof of \Cref{thm:combined}\ref{it:thm-big-eta} is almost the same as the proof of \Cref{thm:combined}\ref{it:thm-small-beta}: we test between linear regression setups with $\sigma {=} 0$ and $\beta {=} (1,\ldots,1)$ vs. $\beta' {=} -\beta$. It suffices to use $\sigma {=} 0$ and $\|\beta\| {=} \sqrt{d}$, since coupling these cases implies a coupling for any scaling $\|\beta\|$ and $\sigma {>} 0$ (see \Cref{app:thm-proofs} for the full proof). Using \Cref{lem:impr-coupling,fact:coupling} completes the proof.\looseness=-1

The remaining lower bound proofs also use \Cref{lem:impr-coupling}. However, since we no longer aim to prove a lower bound of order $\|\beta\|$ (as the upper bound is smaller), we will not use the hypothesis testing setup with $\beta = (1,\ldots,1)$ vs.\ $\beta' = (-1,\ldots,-1)$. Instead, we adjust it as follows. The full proofs for this section are provided in \Cref{app:thm-proofs}.

\paragraph{Proof Sketch of \Cref{thm:combined}\ref{it:thm-interm-eta}} For this lower bound we need a slightly different hypothesis testing problem: Our two hypotheses will use regressors $\beta = s(\eps,\ldots,\eps,1,\ldots,1)$ and $\beta' = s(-\eps,\ldots,-\eps,1,\ldots,1)$ respectively,
where $s > 0$ can be any scaling factor (so that our lower bound holds for every possible norm of $\|\beta\|$) and $\eps\leq 1$ will be chosen shortly. We claim that \Cref{lem:impr-coupling} implies a coupling $\Pi$ between the distributions of \Cref{def:linear_model} for $\beta$ and $\beta'$ such that $\Pr_{((X,y)(X',y')) \sim \Pi}[y \neq y']=0$ and\looseness=-1
\begin{equation}\label{eq:coupling-combined}
     \E\limits_{((X,y),(X',y')) \sim \Pi}\left[\sum_{i=1}^d \1(X_i \neq X_i') \right] \leq 2 + O(\eps\sqrt{d} ) .
\end{equation}
We first complete the theorem proof sketch given \Cref{eq:coupling-combined}. Setting $\eps := (\eta - 4/d)\tfrac{\sqrt{d}}{2C }$ (for $C$ the constant in the big-O), the RHS of \Cref{eq:coupling-combined} becomes $\eta d/2$, making \Cref{fact:coupling} applicable. By the reduction from estimation to testing, this implies every estimator has error $\Omega(\|\beta-\beta'\|)$, which is
    $\|\beta - \beta'\| =  \eps s\sqrt{d/2} = \tfrac{\eps}{\sqrt{1+\eps^2}}\|\beta\| \gtrsim \eps \|\beta\|  
    \gtrsim (\eta - 4/d)\sqrt{d}\|\beta\| \gtrsim \eta \sqrt{d}\|\beta\|$,
where the last step used $\eta \geq 5/d$. In Appendix we use better constants so only require $\eta \geq (2{+}o(1))/d$.

We now prove \Cref{eq:coupling-combined}. By scaling, it suffices to do this only for $s=1$ and $\sigma=0$ (i.e., no additive noise). The coupling is the following: (i) Sample $z$ from the distribution of labels (which is $\cN(0,\sigma^2 + \|\beta\|^2)= \cN(0,(1+\eps^2)d/2)$ for both hypotheses) and set $y=y'=z$. (ii) Sample $t$ and $t'$ from the distribution of the sum of the first half of the coordinates for the two hypotheses. By using standard facts for Gaussians, it turns out that these two distributions are $\cN(\eps z/(1+ \eps^2), d/(2(1+\eps^2)),\cN(-\eps z/(1+ \eps^2), d/(2(1+\eps^2))$. We can also choose to have that $t-t' = 2\eps z/(1+ \eps^2)$ always. (iii) Sample the first half of the coordinates using the coupling of \Cref{lem:impr-coupling} such that they sum to $z-\eps t$. (iv) Sample the second half using \Cref{lem:impr-coupling} such that they sum to $z+\eps t'$. The expected disagreements from the two applications of \Cref{lem:impr-coupling} are
\begin{align*}
\E\left[ \textstyle \sum_{i=1}^d \1\left(  X_i {\neq} X_i'\right) \right]
        \leq 2 + O\big(\E_{t,t'}\left[|t {-} t'| \right]  {+} \eps \E_{t,t'}\left[|t {+} t'| \right]  \big). 
    \end{align*}
    Using $t-t' = 2\eps z/(1 + \eps^2)$ with $z \sim \cN(0,(1{+}\eps^2)d/2)$, $\eps \leq 1$ and that $t,t'$ have variance $O(d)$ proves \Cref{eq:coupling-combined}\looseness=-1.\qed

The constant additive term in \Cref{eq:coupling-combined} breaks the proof of \Cref{thm:combined}\ref{it:thm-interm-eta} when $\eta < 2/d$. To handle this, we improve the coupling below. Roughly speaking, the idea is to break the variables we want to couple into smaller parts and couple the parts separately. When the first parts match, we force the rest to match as well, which (in expectation) reduces the coordinate-wise disagreements. 

\paragraph{Proof Sketch of \Cref{thm:combined}\ref{it:thm-small-eta}} Let $B,E \in (0,1)$ be parameters, let $\1_{d}\in \R^d$ denote the all-ones vector, and define the problem of distinguishing between regressor $\beta := (\tfrac{B}{\sqrt{d/2}}\1_{d/2}, \tfrac{E}{\sqrt{d/2}}\1_{d/2})$ and $\beta' := (\tfrac{B}{\sqrt{d/2}}\1_{d/2}, \tfrac{-E}{\sqrt{d/2}}\1_{d/2})$. The claim now is existence of a coupling $\Pi$ between labeled examples of the linear regression models of \Cref{def:linear_model} with $\beta$ and $\beta'$ such that, $\Pr_{((X,y)(X',y')) \sim \Pi}[y \neq y']=0$  and\looseness=-1
\begin{align}\label{eq:coupling-combined-2}
     \E_{((X,y),(X',y')) \sim \Pi}\left[\sum_{i=1}^d \1(X_i \neq X_i') \right] \lesssim  \frac{E}{\sigma} + \sqrt{d} \frac{E}{B} .
\end{align}
As before, we aim to apply \Cref{fact:coupling} to conclude that the hypothesis testing problem is hard. To do that, we need to ensure that the RHS above is at most $\eta d/2$. Equivalently $E/\sigma \ll \eta d$ and $\sqrt{d}E /B \ll \eta d$. We thus choose $E:= \frac{1}{2C } \min(\eta d \sigma, \eta \sqrt{d}B)$. Noting that $\|\beta - \beta'\| \gtrsim \min(\eta d \sigma, \eta \sqrt{d} \|\beta\|)$ concludes the proof of \Cref{thm:combined}\ref{it:thm-small-eta}. \looseness=-1

It remains to argue \Cref{eq:coupling-combined-2}. The coupling for that is the following. We will use the notation $X=(X_1,X_2)$ to denote the first and second half of coordinates of $X$ and $S_1 = \1_{d/2}^\top X_1 /\sqrt{d/2}$, $S_2 = \1_{d/2}^\top X_2 /\sqrt{d/2}$ to denote the scaled sums of these coordinates (and we will use similar primed letters notation for the corresponding parts of $X'$). The coupling is the following: (i) Draw $y$ from the distribution of labels and let $y'=y$. (ii) Draw $(S_2,S_2')$ from a maximal coupling between the distributions under the null and alternative hypotheses conditioned on $y$. (iii) If $S_2 = S_2’$: draw $X_2$ from the null distribution conditioned on the value of $S_2$ from the previous step, and set $X_2’=X_2$. Draw $(S_1,S_1’)$ from a maximal coupling between null and alternative distributions conditioned on $(S_2, y)$.
     (iii a) If  $S_1=S_1’$ draw $X_1$ from the null distribution conditioned on $S_1$ and set $X_1’=X_1$.
     (iii b) Otherwise use \Cref{lem:impr-coupling} to jointly sample $(X_1,X_1’)$ conditioned on the values of $S_1$ and $S_1’$.
(iv) If $S_2 \neq S_2’$: couple the conditional distributions $S_1$ under the null and $S_1’$ under the alternative given $(S_2, y)$ and $(S_2', y)$, respectively. Then use \Cref{lem:impr-coupling} to jointly sample $(X_2, X_2')$ conditioned on $S_2, S_2'$, and similarly for $(X_1, X_1')$ conditioned on $S_1, S_1'$. 

It is true (and straightforward) that this is indeed a coupling, since the labeled examples are constructed in steps and each step uses the correct conditional distributions under the null/alternative hypotheses, conditioned on all previous steps.

For the disagreement analysis, one must derive the precise form of all the conditional distributions mentioned above and apply \Cref{fact:maximal,lem:impr-coupling} appropriately. This requires some technical work, which we defer to \Cref{app:thm-proofs}. At a high level however, the key point is that the case-splitting in the coupling construction implies that in some cases there are no disagreements, and \Cref{lem:impr-coupling} is applied only to the remaining cases, which occur with probability strictly less than $1$. As a result, the additive $1$ term in the disagreement bound from \Cref{lem:impr-coupling} is scaled by the probability of the corresponding case. Ultimately, this ensures that the expected total number of disagreements is strictly less than $1$, which in turn allows our theorem to apply in the regime $\eta < 1/d$.

\section*{Acknowledgments} The authors would like to thank the anonymous COLT reviewer for identifying a minor bug in an earlier version of this paper.

\newpage

\bibliography{refs}
\newpage

\newpage

\appendix
\section*{Appendix}

The structure of the appendix is as follows:
In \Cref{app:intro}, we provide material omitted 
from \Cref{sec:intro} of the main body, 
including the detailed description of 
upper bounds for our learning task, 
implications for other corruption models, 
a detailed summary of related work, and open problems. 
\Cref{app:prelims} records the notation and mathematical 
background required for our technical results. 
Finally, \Cref{app:warm-up} and \Cref{app:thm-proofs} 
give the technical proofs omitted from \Cref{sec:warm-up} 
and \Cref{sec:thm-proofs} respectively.

\section{Omitted Details from \Cref{sec:intro}} \label{app:intro}

\subsection{Discussion on \Cref{fact:known-algorithms}}
We provide a more detailed presentation of the three algorithms for linear-regression under coordinate-wise corruptions with the specific references to prior robust statistics literature below.

\KNOWNALGORITHMS*
\begin{proof}
    For the first algorithm, one can use any linear regression algorithm developed in the robust statistics literature, with the guarantee that if an $\eps$-fraction of the samples are corrupted (i.e., contain at least one corrupted coordinate), the algorithm's error is $O(\eps \sigma)$. One such polynomial-time algorithm, which uses $\tilde{O}(d/\eps^2)$ samples and works for any $\eps \in (0, \eps_0)$ for a sufficiently small absolute constant $\eps_0$, is the one from \cite{diakonikolas2023near}. Our coordinate-wise contamination model can be reduced to the corruption model for which that algorithm is designed by using $\eps = \eta (d+1)$.
    
    For the second algorithm, observe that for $(X, y)$ drawn from \Cref{def:linear_model}, we have $\E[yX] = \beta$. In particular, for each coordinate $j$, the expectation $\E[yX_j]$ gives an unbiased estimator for the $j$-th coordinate of $\beta$. Under our corruption model, the samples $\{ y_i X_i^{(j)} \}_{i=1}^n$ contain at most a $2\eta$ fraction of corruptions: up to $\eta$ from the $\{X_i^{(j)}\}_{i=1}^n$ and up to $\eta$ from the labels. Moreover, the distribution of the clean samples has sub-exponential tails and variance bounded by $O(\sqrt{\|\beta\|^2 + \sigma^2})$. For such random variables with an $O(\eta)$-fraction of corruptions, it is known that classical estimators such as trimmed mean (cf.~\Cref{fact:trimmed_mean_sub_exp}), yield an error of $O\left(\eta \log(1/\eta) \sqrt{\|\beta\|^2 + \sigma^2}\right)$ using $\tilde{O}(d/\eta^2)$ samples. Due to the additional Gaussian structure (i.e., the fact that $yX$ is not an arbitrary sub-exponential variable, but the product of Gaussian random variables), the usual $\log(1/\eta)$ factor can be removed with a more refined analysis \cite[Theorem 9]{diakonikolas2018robustly} (with sample complexity $O(\poly(d/\eta))$).
    
    Finally, the third algorithm is the trivial one that always outputs the zero vector.
\end{proof}

\subsection{Derivation of Regimes in \Cref{table:combined}}\label{sec:partition}

With \Cref{fact:known-algorithms} in hand, we discuss how the parameter space can be partitioned based on which algorithm from \Cref{fact:known-algorithms} achieves the best error (up to absolute constant factors). Due to space constraints in \Cref{sec:results}, the regimes were summarized there in a single, somewhat condensed table (\Cref{table:combined}). In this section, we present a more detailed version of this partitioning in \Cref{tab:eta-main,table:medium-eta,table:small-eta} and explicitly explain how the different regimes are derived. \Cref{table:combined} in \Cref{sec:intro} is a combined version of \Cref{tab:eta-main,table:medium-eta,table:small-eta} into a single table.

\begin{table}[ht]

  \begin{minipage}{\linewidth}
    \centering
    \begin{tabular}{lcccc}
\toprule
Regime for $\eta$ & $0 \leq \eta < \frac{0.49}{d}$ & $\frac{0.49}{d} \leq \eta < \frac{7}{\sqrt{d}}$ & $\frac{7}{\sqrt{d}} \leq \eta \leq 1$ \\
\midrule
Best upper bound & see \Cref{table:small-eta} & see \Cref{table:medium-eta} & $O(\|\beta\|)$ ($\cA_3$ from \Cref{fact:known-algorithms}) \\
Best lower bound & see \Cref{table:small-eta} & see \Cref{table:medium-eta} & $\Omega(\|\beta\|)$ (\Cref{thm:combined}\ref{it:thm-big-eta}) \\
\bottomrule
\end{tabular}
\vspace{-5pt}
    \caption{Estimation bounds for $\eta$ regimes ($C$ denotes a large constant).}
    \label{tab:eta-main}
  \end{minipage}

  \vspace{1em}

  \begin{minipage}{\linewidth}
    \centering
    \begin{tabular}{lccc}
\toprule
Regime for $\beta$
  & $0 \leq \|\beta\| < \eta \sqrt{d} \sigma$ 
  & $\eta \sqrt{d} \sigma \leq \|\beta\| < \sigma$ 
  & $\sigma \leq \|\beta\| < \infty$ \\
\midrule
\multirow{2}{*}{\shortstack{Best upper bound\\from \Cref{fact:known-algorithms}}} 
  & $O(\|\beta\|)$ 
  & $O(\eta \sqrt{d} \sigma)$ 
  & $O(\eta \sqrt{d} \|\beta\|)$ \\
  & (alg. $\cA_3$) 
  & (alg. $\cA_2$) 
  & (alg. $\cA_2$) \\
\midrule
\multirow{2}{*}{\shortstack{Best lower bound\\from \Cref{thm:combined}} }
  & $\Omega(\|\beta\|)$ 
  & $\Omega(\eta \sqrt{d} \sigma)$ 
  & $\Omega(\eta \sqrt{d} \|\beta\|)$ \\
  & (part \ref{it:thm-small-beta}) 
  & (part \ref{it:thm-small-beta}) 
  & (part \ref{it:thm-interm-eta})$^*$ \\
\bottomrule
\end{tabular}
\vspace{-4pt}
    \caption{Sub-regimes of the $\frac{0.49}{d} {\leq} \eta {<} \tfrac{7}{\sqrt{d}}$ case. $^*$Note: \Cref{thm:combined}\ref{it:thm-interm-eta} is only for $ \eta \in [\frac{2+c}{d},\tfrac{7}{\sqrt{d}}]$.}
    \label{table:medium-eta}
  \end{minipage}

  \vspace{1em}

  \begin{minipage}{\linewidth}
    \centering
    \begin{tabular}{lcccc}
\toprule
Regime for $\beta$ & $0 \leq \|\beta\| < \eta \sqrt{d} \sigma$ & $\eta \sqrt{d} \sigma \leq \|\beta\| < \sigma$ & $ \sigma \leq \|\beta\| < \sqrt{d}\sigma$  & $\sqrt{d} \sigma \leq \|\beta\| $ \\
\midrule
\multirow{2}{*}{\shortstack{Best upper bound\\from \Cref{fact:known-algorithms}}} & $O(\|\beta\|)$ & $ O(\eta \sqrt{d} \sigma )$ & $O(\eta \sqrt{d} \|\beta\| )$ &  $O(\eta d \sigma) $ \\
                             & (alg. $\cA_3$) & (alg. $\cA_2$) & (alg. $\cA_2$) & (alg. $\cA_1$) \\
\midrule
\multirow{2}{*}{\shortstack{Best lower bound\\from \Cref{thm:combined}}} & $ \Omega(\|\beta\|)$ &  $\Omega( \eta \sqrt{d} \sigma )$ & $\Omega( \eta \sqrt{d}\|\beta\|  )$  & $\Omega \left( \eta d \sigma \right)$   \\
                             & (part \ref{it:thm-small-beta}) & (part \ref{it:thm-small-beta})  & (part \ref{it:thm-small-eta})  &(part \ref{it:thm-small-eta})   \\
\bottomrule
\end{tabular}
\vspace{-4pt}
    \caption{Sub-regimes of the $0 \leq \eta < \frac{0.49}{d}$ case.}
    \label{table:small-eta}
  \end{minipage}

\end{table}

\paragraph{Regime $7/\sqrt{d} \leq \eta \leq 1$} Algorithm $\cA_1$ from \Cref{fact:known-algorithms} is not applicable in this regime. Among the other two algorithms, $\cA_3$ always has the best error (because the error of $\cA_2$ is in the order of $\eta \sqrt{d}\sqrt{\|\beta\|^2 + \sigma^2}  \geq \eta \sqrt{d} \|\beta\| \geq 7\|\beta\|$).  This explains the last column in \Cref{tab:eta-main}.

\paragraph{Regime $0.49/d \leq \eta < 7/\sqrt{d}$} Again,  algorithm $\cA_1$ is not applicable in this regime. Between $\cA_2$ and $\cA_3$, we note that whenever $\|\beta\| \geq \eta \sqrt{d} \sigma$, the error of $\cA_2$ is $O(\eta \sqrt{d}\sqrt{\|\beta\|^2 + \sigma^2})=O(\eta \sqrt{d}\|\beta\| + \eta \sqrt{d}\sigma) = O(\|\beta\|)$, i.e., $\cA_2$ has the best error (up to constant factors). Moreover, when $\|\beta\| < \sigma$, the term that dominates in the error of $\cA_2$ is $\eta\sqrt{d}\sigma$ and when $\|\beta\| \geq\sigma$ the dominating $\cA_2$ error term is $\eta \sqrt{d} \|\beta\|$. This explains the regimes shown in \Cref{table:medium-eta}.

\paragraph{Regime $0 \leq \eta < 0.49/d$} In this regime, all three algorithms are applicable. For the same reason as before, the error of $\cA_3$ is better than that of $\cA_2$ when $\|\beta\| < \eta \sqrt{d} \sigma$ and $\cA_2$ is better than $\cA_3$ otherwise. It also outperforms $\cA_1$, whose error is on the order of $\eta d \sigma$.
In the sub-regime $\eta \sqrt{d} \sigma \leq \|\beta\| < \sqrt{d} \sigma$, the error of $\cA_2$ is at most $O(\eta \sqrt{d} \sigma + \eta \sqrt{d} \|\beta\|)$, with both terms being smaller than $O(\eta d \sigma)$. Thus, $\cA_2$ achieves the best error in this sub-regime.
In the final sub-regime $\|\beta\| \geq \sqrt{d}\sigma$, the error of $\cA_1$ is of the order $\eta d \sigma$, which is smaller than that of $\cA_3$ (which is of the order $\|\beta\|  \geq \sqrt{d} \sigma \geq \sigma \geq \eta d \sigma$). Comparing with $\cA_2$, the dominating term in the error of $\cA_2$ is the $\eta \sqrt{d} \|\beta\|$ term which is larger than $\eta d \sigma$ whenever $\|\beta\| > \sqrt{d}\sigma$. Thus $\cA_1$ has the best error (up to a constant factor) in the  sub-regime $\|\beta\| > \sqrt{d}\sigma$. This explains \Cref{table:small-eta}.

\subsection{Unified Algorithm with Error Guarantees as in \Cref{tab:eta-main,table:medium-eta,table:small-eta}}\label{sec:unified}

We further claim that there exists a single algorithm that can automatically adapt to the best choice among $\cA_1, \cA_2, \cA_3$ from \Cref{fact:known-algorithms}, and achieve the error upper bounds in \Cref{tab:eta-main,table:medium-eta,table:small-eta}, without requiring a priori knowledge of the regime in which $\|\beta\|$ lies. The algorithm only needs to know $\eta$ (which could also be avoided by applying standard techniques like Lepski's method, but is beyond the scope of this work).

\begin{theorem} \label{thm:unified}
    There is a polynomial time algorithm $\cA$ that takes as input a parameter $\eta \in (0,1)$ and $n = O(\poly(d/\eta))$ samples from the $d$-dimensional linear regression model (\Cref{def:linear_model}) after  $\eta$-fraction of coordinatewise corruptions according to \Cref{def:model2}, and returns an estimate $\hat \beta$ of $\beta$ that satisfies the error upper bounds from \Cref{table:combined} (equivalently \Cref{tab:eta-main,table:medium-eta,table:small-eta}) with high constant probability.
\end{theorem}
\begin{proof}
The idea for the meta-algorithm is to first use the knowledge of $\eta$ to restrict our attention to the relevant table among \Cref{tab:eta-main,table:medium-eta,table:small-eta} and then estimate the error bound of the three algorithms $\cA_1,\cA_2,\cA_3$ (up to absolute constant factors) in order to return the output of the algorithm corresponding to the smallest error bound. We show how this can be done for the regime $0 < \eta \ll 1/d$; the other regimes can be handled by the same arguments.

Suppose that $\eta \ll 1/d$ and let $\hat{\beta}_1, \hat{\beta}_2, \hat{\beta}_3$ be the outputs of $\cA_1,\cA_2$ and $\cA_3$ respectively. We first show how to estimate (within a multiplicative absolute constant) the upper bound on the error of $\cA_3$ from \Cref{table:small-eta} (i.e., the quantity $\eta d \sigma$). Since $\eta$ and $d$ are known, it suffices to estimate $\sigma$. Note that the expectation of the squared labels $y^2$ under the clean distribution of \Cref{def:linear_model} is $\E[y^2] = \| \beta\|^2 + \sigma^2$. The procedure to estimate $\sigma$ within a constant factor is then the following: We compute the residuals $y' = y - \hat{\beta}_1^\top X$ for all samples in the dataset. This transformation makes the clean samples as if they came from the model of \Cref{def:linear_model} with true regressor $\beta - \hat{\beta}_1$ (instead of $\beta$ that we started with). The expectation of the new squared labels is now $\E[(y')^2] = \|\beta - \hat{\beta}_1\|^2 + \sigma^2 = \sigma^2 ( 1 + O(\eta d)) = O(\sigma^2)$. We can use an outlier robust one-dimensional mean estimator to find a $u$ such that $|u - \sigma^2| \leq 0.01 \sigma^2$. This can be done using  \Cref{fact:trimmed_mean_sub_exp} below (where $\|(y')^2\|_{\psi_1}$ in our case is at most $O(\sigma^2)$):

\begin{fact}[Univariate Trimmed Mean, see, e.g., \cite{DK2023}] \label{fact:trimmed_mean_sub_exp}
  Let $\eps_0$ be a sufficiently small absolute constant. 
  There is an algorithm (trimmed mean) that, for every $\eps \in (0,\eps_0)$ and a univariate distribution $D$ that has $\psi_1$-norm at most $\sigma^2$ (cf.\ \Cref{def:orlich}), given a set of $n  \gg \log(1/\delta)/(\eps^2\log^2(1/\eps))$ samples from $D$ with corruption at rate $\eps$, outputs a $\widehat{\mu}$ such that $|\widehat{\mu} - \E_{X \sim D}[X]| \lesssim\sigma^2\eps\log(1/\eps)$ with probability at least $1-\delta$.
\end{fact}

\begin{definition}[Sub-exponential Random Variables]\label{def:orlich}
    We call $\snorm{\psi_1}{Y} := \sup_{p\geq 1}p^{-1} \E[|Y|^p]$ the sub-exponential norm of the random variable $Y$. 
\end{definition}
By using the above, so far we have shown how to obtain an estimate $e_1$ such that $|e_1 - \eta d \sigma| \leq 0.1 \eta d \sigma$ (i.e., $e_1$ is a multiplicative approximation to the error bound  of $\cA_1$ listed in \Cref{table:small-eta}). By using \Cref{fact:trimmed_mean_sub_exp} again we can similarly obtain an estimate for the error bound of $\cA_2$ from the table, i.e., $e_2$ that satisfies $|e_2 - \eta \sqrt{d} \sqrt{\|\beta\|^2 + \sigma^2}| \leq 0.1 \eta \sqrt{d}\sqrt{\|\beta\|^2 + \sigma^2}$.  Let $C,C',C''$ be sufficiently large absolute constants with $C \ll C' \ll C''$. The meta-algorithm that chooses which of the outputs $\hat{\beta}_1,\hat{\beta}_2,\hat{\beta}_3$ to return is the following. If $C e_1 < e_2 $ the meta-algorithm returns $\hat{\beta}_1$, otherwise it does the following: if $\|\hat{\beta}_2\| > C' e_2$  the meta-algorithm returns $\hat{\beta}_2$, otherwise it returns $\hat{\beta}_3$. The correctness follows easily: 
In the regime $\|\beta \| \gg \sqrt{d} \sigma$ where $\cA_1$ has the best error, the first check yields $e_2 = \Omega(\eta \sqrt{d} \sqrt{\sigma^2 + \|\beta\|^2}) = \Omega( \eta \sqrt{d}\|\beta\|) \gg e_1$ thus the meta-algorithm will indeed return $\hat{\beta}_1$. 
If it is the case that that we are in the regime where $\cA_3$ has the best error, i.e., $\|\beta\| < C e_2$  then we can easily see that the meta-algorithm will return $\hat \beta_3$. This is because the first check will yield $e_2 \leq C e_1$ (since $e_2 = O(\eta \sqrt{d}\sqrt{\sigma^2 + \|\beta\|^2}) =O( \eta \sqrt{d} \sigma) = O(\eta d \sigma) \leq C e_1$) and the second check will yield
$\|\hat{\beta}_2\| \leq \|\beta\| + \|\beta - \hat{\beta}_2\| \leq C e_2 + O(\eta \sqrt{d}\sqrt{\|\beta\|^2 + \sigma^2}) \leq C' e_2$. If we are in the regime where $\cA_2$ has the best error, i.e., $\|\beta\| > C'' e_2$ the first check will yield $e_2 \leq C e_1$ (since $e_2 = O(\eta \sqrt{d}\sqrt{\sigma^2 + \|\beta\|^2}) = O(\eta d \sigma) \leq C e_1$) and the second check will yield $\|\hat{\beta}_2\| \geq \|\beta\| - \|\beta - \hat{\beta}_2\| \geq C'' e_2 - O(e_2) >  C'e_2$ thus the meta-algorithm will indeed yield $\hat{\beta}_2$. Note that we left out the regime $C' e_2 \leq \|\beta\| \leq C'' e_2$. In that case, it does not matter which of $\hat{\beta}_2,\hat{\beta}_3$ we output because their errors are within a constant factor from each other.

\end{proof}

\subsection{Other Implications of Our Results} \label{app:discussion}

In this brief subsection, we discuss some immediate implications of our results.

\paragraph{Implications on Models with a Restricted Set of Corrupted Coordinates}

In this paper, we primarily discuss the case in which every coordinate is subject to an $\eta$ fraction of corruptions. However, another reasonable model would be to consider the case where only a subset $\mathcal{S} \subset \{1, 2, \dots, d\}$ of $|\mathcal{S}| = r \leq d$ coordinates is subject to adversarial erasures as in \Cref{def:model} or more general corruptions as in \Cref{def:model2}. It is straightforward to verify that all our results generalize to this model, with estimation error bounds now depending on $r$ in place of $d$ and on $\|\beta_{\mathcal{S}}\|_2$ in place of $\|\beta\|_2,$ where $\beta_{\mathcal{S}} = \sum_{i\in \mathcal{S}}\beta_i e_i$ and $e_i$ denotes the $i^{\rm th}$ standard basis vector. 

\paragraph{Implications for Generalized Linear Models}
While our primary focus in this work is on lower and upper bounds for the attainable estimation error $\|\beta - \hat{\beta}\|_2$, some of our lower bounds have broader implications on the attainable squared error of generalized linear models (GLMs)---where the ``ground truth'' labels follow $f(\beta^\top X)$ for some possibly more general function $f$ called the activation,  instead of $\beta^\top X$ corresponding to the case of linear regression (where $f(t) =t$ is the identity function). In this case, the focus is usually on bounding the mean squared error for a predictor $\hat{\beta}$, defined by $\mathcal{L}(\hat{\beta}) = \E_{(X, y)\sim \mathcal{D}}[(y - f(\hat{\beta}^\top X))^2]$ for $(X, y)$ jointly distributed according to a distribution $\mathcal{D}$. It is immediate that in the case of linear regression model from \Cref{def:linear_model}, we have 
\begin{align*}
    \mathcal{L}(\hat{\beta}) &= \sigma^2 + \E_{X\sim \mathcal{N}}[(\beta^\top X - \hat{\beta}^\top X)^2]\\
    &= \sigma^2 + \E_{X\sim \mathcal{N}}[(\beta - \hat{\beta})^\top XX^\top (\beta - \hat{\beta})]\\
    &= \sigma^2 + \|\beta - \hat{\beta}\|_2^2.
\end{align*}
Thus, our lower bounds on $\|\beta - \hat{\beta}\|_2$ directly imply lower bounds on the mean squared error for linear regression, using the above derivation.

Concerning GLMs in broader generality, let us consider the realizable case, where all labeled examples $(X, y)$ satisfy $y = f(\beta^\top X)$ for some fixed vector $\beta$ and activation function $f.$  
Notice first that the lower bounds in \Cref{thm:combined}\ref{it:thm-big-eta} and \Cref{thm:combined}\ref{it:thm-interm-eta} apply even when $\sigma = 0$.\footnote{While constructions of other lower bounds (from \Cref{thm:combined}\ref{it:thm-small-beta},\ref{it:thm-small-eta}) may also have implications on the mean squared error of GLMs for $\sigma \neq 0$, here we focus on \Cref{thm:combined}\ref{it:thm-big-eta},\ref{it:thm-interm-eta} as the implications are more direct.} Our lower bounds for linear regression are proved by constructing a coupling between the distribution of $X$ conditioned on the value of $y=t$ under two distinct hypotheses corresponding to sufficiently different prediction/weight vectors $\beta$ and $\beta'$. Since for $\sigma = 0$ in the case of linear regression this is equivalent to conditioning on the value of $\beta^\top X = (\beta')^\top X = t$, the same lower bounds (implying impossibility of distinguishing between $\beta$ and $\beta'$) generalize to the case where $y = f(t)$. This is because considering more general functions $f$ that are not identity (as in the case of linear regression) can only further obscure information if $f$ is not invertible (as is the case, for example, for ReLU activations where $f(t) = \max\{0, t\}$). 

The above discussion implies that the same lower bounds as in \Cref{thm:combined}\ref{it:thm-big-eta} and \Cref{thm:combined}\ref{it:thm-interm-eta} apply for the value of $\|\hat{\beta} - \beta\|_2$, where $\beta$ is the ``ground truth'' predictor and $\hat{\beta}$ the predictor that can be constructed by an algorithm. A consequence is that the mean squared error can be bounded below as a function of $\|\hat{\beta} - \beta\|_2^2$ for a broad class of GLMs. In particular, for the class of $(a, b)$-unbounded activations---namely, monotonically non-decreasing functions $f$ that are $b$-Lipschitz, satisfy $f(0) = 0,$ and are such that $f'(t) \geq a > 0$ for $t > 0$---known results (cf.\ \cite{LKDD2024,WZDD2023}) imply that for $X$ distributed according to the standard normal distribution, we have $\mathcal{L}(\hat{\beta}) = \Omega(\mathrm{poly}(b/a))\|\hat{\beta} - \beta\|_2^2.$ In particular, since ReLU activation is $(a, b)$-unbounded with $a = b = 1,$ we have that for the ReLU activation, $\mathcal{L}(\hat{\beta}) = C\|\beta - \hat{\beta}\|_2^2$ for a universal constant $C > 0.$ As a consequence, the mean squared error is of the order $\Omega(\|\beta\|_2^2)$ for $\eta \in [7/\sqrt{d}, 1]$ and of the order $\Omega(c^2\eta^2 d \|\beta\|_2^2)$ for $\eta \in [(2+c)/d, 7/\sqrt{d}].$ In other words, unless both $\eta$ and $\|\beta\|_2$ are small, no meaningful prediction---as measured by the mean squared error---is possible. 

These observations reinforce the argument that linear regression is the most basic model to study when considering limitations imposed by adversarial coordinate-wise corruptions studied in our work.

\subsection{Related Work} \label{app:related_work}

There is a vast literature on learning with missing data in various forms. To the best of our knowledge, this work is the first to address the fully adversarial missingness model of \Cref{def:model} in the context of high-dimensional regression. The most relevant prior work is \cite{LPRT2021}, which considers fully adversarial missingness, but for the simpler task of mean estimation. \cite{HR2021} also considers mean estimation and uses a milder setting for the missingness model. We provide a discussion of these works below, followed by a broader (though not exhaustive) review of related literature, with an emphasis on settings most closely aligned with ours.

\paragraph{Comparison with \cite{LPRT2021}}
This work studies the optimal error of Gaussian mean estimation in the same coordinate-wise adversary of \Cref{def:model,def:model2} (among others). Concretely, it considers a slightly more general problem than the vanilla Gaussian mean estimation where data has some additional low-dimensional structure: each sample $x_i$ follows $\cN(\mu,\Sigma)$ but can be written as $x_i = A z_i$, where $A \in \R^{d \times r}$ and $z \in \R^{r}$ is some lower dimensional Gaussian. The goal is to estimate $\mu$. Their lower bounds on the estimation error are based on the approach of constructing couplings with small coordinate-wise disagreements (as in the present work). In terms of algorithms, the approach of \cite{LPRT2021} for the missing data of \Cref{def:model} is a two-step algorithm that first estimates the missing values and then runs existing estimators on the imputed dataset. They show that the first step becomes NP-hard when data is replaced (\Cref{def:model2}) instead of missing. However, they are able to get some preliminary results by a randomized algorithm for special cases (like when $A$ is known).
In the most general case for the mean estimation problem of \cite{LPRT2021}, the answer to the question of whether missing data are easier to handle than replaced data remains unclear. In the case of the linear regression problem considered in this paper, we find that the answer is negative.
On the technical side, our work builds on the coupling technique of \cite{LPRT2021}, but significantly extends it to handle new challenges in the linear regression setting. The main difficulty is that the estimation error behaves differently across many distinct regimes, each requiring a separate analysis. As discussed in \Cref{sec:overview}, a direct application of the \cite{LPRT2021} coupling fails when the additive noise $\sigma$ is small, leaving several regimes unaddressed. We resolve this by coupling only the first $d-1$ coordinates and treating the last one as additive noise.
This regime also requires a redesigned hypothesis testing setup. The original test inspired from \cite{LPRT2021} that compares $\beta = (1,\ldots,1)$ and $\beta' = (-1,\ldots,-1)$, is tailored to show an $\Omega(\|\beta\|)$ lower bound, i.e., that no non-trivial estimation is possible. In our case, however, some non-trivial estimation is possible, so we design a finer test: we split the coordinates into two halves with differing values in $\beta$ and $\beta'$ and apply separate couplings to each half.
Finally, one particularly challenging regime remains, which we address by further splitting both the coordinates and the labels into two parts, and using a more delicate coupling between the parts of covariates and labels.

\paragraph{Comparison with \cite{HR2021}}
This work also studies mean estimation under a missing data model. Their setting involves a combination of two types of erasures: up to an $\eta$-fraction of coordinates may be missing, and up to an $\eps$-fraction of entire samples may be deleted. The second type of deletion (entire samples) can be fully adversarial, while the first type (coordinate-wise) is less adversarial than in our model, because in \cite{HR2021} the adversary must commit to the missingness pattern before seeing the data. The algorithms proposed in \cite{HR2021} are based on imputing missing values, and thus do not naturally extend to the contamination setting, where adversaries may replace values rather than erase them.

\paragraph{Literature on Missing Data} The study of missing data has a long history, motivated by concerns like those discussed in \Cref{sec:intro}. In a seminal work, \cite{Rub1976a} distinguishes three types of missingness: when missingness is independent of the data, it is called missing completely at random (MCAR); when it depends only on the observed data, it is missing at random (MAR); otherwise, it is missing not at random (MNAR), where missingness depends on both observed and unobserved values. See also \cite{Tsi2006,TLR2003,LR2019,MP2021} for a more in-depth review.

Much of the literature on regression with missing data has focused on parameter estimation. Early works include \cite{Lit1992, Lit1993}, as well as \cite{RT2010}, which proposes a sparse variant of the Lasso estimator. More recent approaches revisit the problem from the perspectives of imputation and collaborative learning \cite{CAM2020, CCD2023}.

Other recent efforts have shifted focus to the label prediction problem, which is a fundamentally different goal, as the test set is also expected to contain missing entries. Notably, the Bayes optimal predictor in this setting decomposes into a sum over predictors for each missingness pattern, and the number of such patterns can be exponential. This pattern-specific structure has been studied in \cite{MPJ+2020, LJM+2020, ABDS2022}, which characterize minimax-optimal rates under MAR and MCAR (and in \cite{SBC2024} for classification). A body of work also focuses on the idea of using a two-step procedure: first impute the missing entries, then apply algorithms designed for complete data. This method has been shown to be Bayes optimal in various settings \cite{JCP+2024, bertsimas2024simple, LJSV2021}.

In light of the above, our work fits into the MNAR setting, but to the best of our knowledge, is not directly related to any of the aforementioned works. Unlike the prediction-focused literature (which is more challenging and often incurs exponential complexity) we study parameter estimation. Moreover, prior work on parameter estimation for linear regression either assumes MAR \cite{LW2012, CCD2023} or considers much more benign forms of MNAR; for instance, \cite{CAM2020} assumes that the missingness pattern is drawn i.i.d. from an arbitrary distribution.

\paragraph{Truncated Statistics}
Although technically this literature falls into the Missing Not At Random category from the previous paragraph, it is sufficiently well developed to have its own place in the literature. Truncated statistics refer to the situation where samples falling outside of a fixed set are censored. Parameter estimation under this setting traces back to the early work of Galton, Pearson and Lee \cite{Gal1897,Pea1902,PL1908}. Despite early interest in the problem, computationally efficient algorithms for multivariate Gaussians were developed relatively recently, with some notable works being \cite{DGTZ2018,KTZ2019} for mean estimation and  \cite{DGTZ2019,DRZ2020,DSYZ2021} for linear regression. While the previous works operate under a setting where each sample is either entirely censored or entirely visible, \cite{BDGW2025} extends the setting to partially truncated data, which each coordinate individually visible or censored (for the problem of Gaussian mean estimation). That said, the setting considered 
in \cite{BDGW2025} is still milder than ours. It considers either (i) a setting where each coordinate is erased based on membership in a fixed, predetermined set (to which the algorithm has oracle access), or (ii) erasures based on a simple projection rule. Under this model, consistent estimation is possible, as shown in \cite{BDGW2025}, whereas in our fully adversarial erasure setting, it is provably not.

\paragraph{Robust Statistics}
Robust Statistics was developed in the 1960s \cite{Tuk60, Hub64} to handle scenarios where a small fraction of the data points are arbitrarily corrupted. As in the setting of this work, such  corruptions do not allow for consistent estimation. Early work focused on characterizing the information-theoretic error for univariate Gaussian mean estimation. Since then, a large body of work developed estimators for various tasks \cite{HR2009}. However, these estimators were computationally inefficient for high-dimensional tasks. The field saw a resurgence with the work of \cite{LRV2016,DKK+19} that tackled high-dimensional robust mean estimation in polynomial time. 
In the last decade, a plethora of robust estimators have been introduced for different tasks, including linear regression \cite{KKM2018,DKS2019,PJL2024,CAT+2020}  considered in this work. For a comprehensive treatment, see the book \cite{DK2023}. The key difference is that the de facto contamination models in robust statistics treat each sample as either fully clean or fully corrupted, making non-trivial estimation impossible as the corruption rate approaches $1/2$. This fails to capture cases where each coordinate is only mildly corrupted, i.e., few samples are corrupted per coordinate, even if most samples are affected overall. Such settings may still allow for meaningful estimation, but algorithms from the robust statistics literature are inapplicable.

\subsection{Open Problems}\label{sec:open-problems}

While our work addresses the estimation error for the standard model of linear regression with coordinate-wise corruptions by providing matching upper and lower bounds for essentially all possible parameter regimes, there are several avenues for future research that merit further investigation.   

First, it would be interesting to consider $X \sim \cN(0,\Sigma)$ in the linear regression problem, with $\Sigma \neq I$ and unknown to the algorithm in the linear regression model. Does the conclusion that dealing with missing data is no easier than handling adversarially replaced data still hold in this setting?

Second, although the algorithms are efficient in terms of runtime (i.e., polynomial time) and sample complexity (using $\tilde O(d/\eta^2)$ samples), the optimal sample complexity for the problem remains unclear. In the clean setting (where all samples follow the distribution of \Cref{def:linear_model}), the optimal sample complexity is $d/u^2$, where $u$ is the target estimation error. In some regimes of \Cref{tab:eta-main,table:medium-eta,table:small-eta}, the estimation error scales with $\|\beta\|$, which can be very large. This implies that (at least in the noiseless setting) fewer samples are needed to achieve that level of error.

Third, as discussed in \Cref{app:discussion}, some of our lower bounds directly imply lower bounds for the mean squared error of a broad class of generalized linear models that includes ReLU as a special case, even with perfect, noise-free labels (i.e., in the realizable case). It is however unclear if those lower bounds are tight or if they can be strengthened further, since non-invertible activations like ReLU can further obscure information available to the algorithm (because the label is zero whenever the argument $\beta^\top X$ is negative). Relatedly, it would be interesting to develop polynomial-time algorithms that can match the information-theoretic lower bounds (or prove no such algorithms exist, even when restricted to classes such as Statistical Query algorithms).

Finally, our results demonstrate that adversarial coordinate-wise deletions make both estimation and prediction challenging, even in the most basic model of Gaussian linear regression and with infinite samples. On the other hand, a fully random model of deletions (where deletions are independent of the observed data) allows for the diminishing estimation error as the sample size is increased. These two extremes beg the question of what lies in between, when the deletions are neither fully adversarial nor fully random.

\section{Preliminaries}\label{app:prelims}

We provide the full version of the preliminaries here.

\subsection{Notation}

\paragraph{Basic Probability Notation} We write $X \sim D$ to denote a random variable $X$ that is distributed according to the distribution $D$. We use $P_X(x)$ to denote the pdf of $X$. For multiple variables, we use $(X,Y) \sim D$ to denote that $X$ and $Y$ are jointly distributed according to $D$, and write $P_{X,Y}(x,y)$ for the pdf of that joint distribution. If $(X,Y) \sim D$ we will also use the notation $X|Y=y$ to denote the random variable distributed according to the \emph{conditional} distribution of $X$ given the occurrence of the value $y$ for $Y$, i.e., the distribution whose pdf is the conditional density function $P_{X|Y}(x|y) = P_{X,Y}(x,y)/P_Y(y)$, and we denote by $X$ the random variable distributed according to the \emph{marginal} distribution of $D$, i.e., the distribution with pdf $P_X(x) = \int P_{X,Y}(x,y) \d y$.
We use $\E_{X \sim D}[X]$ for the expectation of $X$.
If $X \sim P_X$ and $Y \sim P_Y$ are continuous random variables over some domain $\mathcal{X}$, with pdfs $P_X(x)$ and $P_Y(y)$, we denote by $\dtv(X,Y) = \frac{1}{2} \int_{x \in \mathcal{X}} |P_X(x) - P_Y(x)| \d x$ the total-variation distance between the distributions $P_X$ and $P_Y$ (by slightly abusing notation, $\dtv(X,Y)$ may some times be also denoted as $\dtv(P_X,P_Y)$). 
We denote the Kullback–Leibler (KL) divergence between distributions $X \sim P_X$ and $Y \sim P_Y$ by $\DKL(X  \,\|\, Y) = \int_{x \in \mathcal{X}} P_X(x) \log\left(\frac{P_X(x)}{P_Y(x)}\right) \d x$, assuming $P_X$ is absolutely continuous with respect to $P_Y$.

\paragraph{Basic Notation}  We use $\mathbb{Z}_+$ for the set of positive integers. We denote $[n]=\{1,\ldots,n\}$. For a vector $x$ we denote by $\| x \|$ its  Euclidean norm. Let $I_d$  denote the $d\times d$ identity matrix (omitting the subscript when it is clear from the context). We use $\1_d$ for the all-ones vector in $\R^d$.
 We use  $\top$ for the transpose of matrices and vectors.
 We use $|A|$ to denote the determinant of matrix $A$.
We use $a\lesssim b$ to denote that there exists an absolute universal constant $C>0$ (independent of the variables or parameters on which $a$ and $b$ depend) such that $a\le Cb$. In our notation $a = O(b)$ has the same meaning as $a \lesssim b$ (similarly for $\Omega(\cdot)$ notation) We use $\tilde O$ and $\tilde \Omega$ to hide polylogarithmic factors.

\paragraph{Couplings} If $P$ and $Q$ are probability distributions over $\cX$, then a coupling $\Pi$ of $P$ and $Q$ is any distribution over $\cX \times \cX$ such that the marginals of $\Pi$ coincide with $P$ and $Q$. We write $(X,Y) \sim \Pi$ to denote random variables that are distributed according to the coupling; $X$ is marginally distributed according to $P$ and $Y$ according to $Q$.

\subsection{Useful Probability and Linear Algebraic Facts}

We start with some facts about Gaussian distributions, namely the inner product of a Gaussian and a fixed vector is another Gaussian, and the distribution conditioned on one of the variables is also Gaussian.

\begin{restatable}[see, e.g., \cite{PP2008}]{fact}{FACTLINEAR}\label{fact:linear}
    If $X \sim \cN(\mu, \Sigma)$ is a multivariate Gaussian vector in $\R^d$, and $u \in \R^d$ is another fixed (deterministic) vector, then $u^\top X \sim \cN(u^\top \mu, u^\top\Sigma u)$.
\end{restatable}

\FACTCONDITIONAL*

The following fact provides a closed-form formula for the KL-divergence between two multivariate Gaussians. The formula can be derived by direct computation and using properties from Section 8.2 of \cite{PP2008}.

\begin{restatable}[KL-divergence between multivariate Gaussians]{fact}{FACTKLGAUSSIANS}\label{fact:KL-Gaussians}
    The KL-divergence between $\cN(\mu_1, \Sigma_1)$ and $\cN(\mu_2, \Sigma_2)$ is
    \begin{align*}
        \DKL\left(\cN(\mu_1, \Sigma_1)  \,\|\,  \cN(\mu_2, \Sigma_2)  \right)
         = \frac{1}{2} \left( 
\log \frac{|\Sigma_2|}{|\Sigma_1|} 
- d 
+ \operatorname{tr}(\Sigma_2^{-1} \Sigma_1) 
+ (\mu_2 - \mu_1)^\top \Sigma_2^{-1} (\mu_2 - \mu_1) 
\right) .
    \end{align*}
\end{restatable}

We will bound total variation distances using the previous fact combined with Pinsker's inequality, stated below.

\begin{restatable}[Pinsker's Inequality (see, e.g., \cite{Tsy2008})]{fact}{FACTPINSKER}\label{fact:pinsker}
Let $P$ and $Q$ be two probability distributions over the same measurable space. Then,
\[
\DTV(P, Q) \leq \sqrt{\frac{1}{2} \DKL(P  \,\|\, Q)}.
\]
\end{restatable}

Combining the above two facts we obtain the following bound on the TV distance between univariate Gaussians.
\FACTTVGAUSSIANS*

The following fact states that there exists a coupling between two distributions such that the probability of disagreement is at most their total variation distance.

\FACTMAXIMAL*

We also require the multiplicative version of the Chernoff-Hoeffding bound for binary random variables:

\begin{fact}[Chernoff-Hoeffding Bound (see, e.g., \cite{DP2009})]\label{fact:Chernoff}
    Let $g_1,\ldots,g_n$ be random variables in $\{0,1\}$ such that $\E[g_i] = p$ for all $i \in [n]$. Then, for all $\eps \in (0,1)$ the following holds:
    \begin{align*}
        \Pr\left[ \frac{1}{n} \sum_{i=1}^n g_i > p(1+\eps) \right] \leq \exp\left( -\frac{\eps^2}{3}pn \right) .
    \end{align*}
\end{fact}

Finally, the following linear algebraic fact provides a useful formula for inverting matrices of the form identity minus a rank-one matrix.

\begin{restatable}[Sherman-Morrison formula (see, e.g., Fact 3.21.3 in \cite{Ber2018})]{fact}{FACTSM}\label{cl:sherman-morrison}
    The matrix $\Sigma = I + u v^\top$ is invertible if and only if $1 + v^\top u \neq 0$. In this case, $\Sigma^{-1} = I - \frac{uv^\top}{1 + v^\top u}$
\end{restatable}

\section{Omitted Details from \Cref{sec:warm-up}}\label{app:warm-up}

\subsection{Couplings with Small Coordinate-Wise Disagreements}\label{app:warm-up-couplings}
We restate and prove the following statements.

\FACTCOUPLING*
\begin{proof}
    Given a coupling that satisfies \Cref{eq:coupling-assumption}, we describe below a procedure that generates samples for each hypothesis and uses a simple adversary to edit the samples so that the resulting data set is (with high probability) the same regardless of the hypothesis in effect. The procedure consists of simply drawing paired samples from the coupling and the adversary erases all coordinates where the samples differ. 
    
    \vspace{10pt}
    \begin{mdframed}
    \begin{enumerate}
        \item Let $\Pi$ denote a coupling that satisfies \Cref{eq:coupling-assumption}.
        \item Initialize empty sets $S \gets \emptyset$, $S' \gets \emptyset$.
        \item Initialize corruption budgets $r_1 \gets \eta \, n, \ldots,  r_{d+1} \gets \eta \, n$ for each of the $d$ coordinates as well as the labels.
        \item For $i=1,2,\ldots,n$ do:
        \begin{enumerate}
            \item Draw $((X,y),(X',y')) \sim \Pi$.
            \item For every $j =1,2,\ldots, d+1$
            \begin{enumerate}
                \item If $j \leq d$ and $X_j \neq X_j'$ and $r_j > 0$:\label{line:disagree1}
                \begin{enumerate}
                    \item $X_j \gets \perp$, and $X_j' \gets \perp$. \label{line:delete1}
                    \item Update $r_j \gets r_j - 1$.
                \end{enumerate}
                \item If $j = d$ and $y \neq y'$ and $r_j > 0$:\label{line:disagree2}
                \begin{enumerate}
                    \item $y \gets \perp$, and $y' \gets \perp$. \label{line:delete2}
                    \item Update $r_j \gets r_j - 1$.
                \end{enumerate}
                \item $S \gets S \cup \{(X,y)\}$
                \item $S' \gets S '\cup \{(X',y')\}$
            \end{enumerate}
        \end{enumerate}
    \end{enumerate}
\end{mdframed}
 \vspace{10pt}

    Each time the adversary above deletes the $j$-th coordinate (lines \ref{line:delete1}, \ref{line:delete2}) it reduces the budget $r_j$ by 1. If the budget reaches zero, the adversary can no longer keep deleting that coordinate.
    Importantly, if $\cE$ denotes the event that none of the $r_j$'s for $j=1,2\ldots,d+1$ reach zero, the dataset $S$ at the end is the same regardless of the hypothesis that is under effect. This means that under that event, no algorithm can distinguish between the two hypotheses. 
    It remains to show that the probability of the event $\cE$ is at least $1- (d+1) e^{- \Omega(c^2 n \eta)}$. For simplicity let us first focus on the corruptions of covariates only (and will discuss label corruptions at the end). That is, let $\cE_1$ denote the event that none of the $r_j$'s for $j=1,2\ldots,d$ reach zero, and we will show that this happens with probability at least $1- d e^{- \Omega(c^2 n \eta)}$. A similar argument will work for the labels thus we only focus on showing that the dataset restricted to the covariates becomes indistinguishable.

    Let $g_{i,j} \in \{0,1\}$ be $1$ if and only if lines \ref{line:delete1}, \ref{line:delete2} caused a deletion of the $j$-th coordinate of the $i$-th sample during the process described above. By assumption we have that for every $i\in [n]$, 
    \begin{align*}
        \E\left[\sum_{j=1}^{d } g_{i,j} \right]  \leq \frac{\eta}{2} d (1-c). 
    \end{align*}
    By the assumption in the lemma statement, we can assume that the first half of the coordinates undergo a random permutation and that the second half of the coordinates also undergo another random permutation.
    If we denote by $g_{i,j}' \in \{0,1\}$ the random variable that is $1$ if and only if the $i$-th sample has its $j$-th coordinate corrupted after the aforementioned two random permutations, then we have the following for every $i\in [n]$ and $j \in [d/2]$ (i.e., we are only analyzing the first half of coordinates for now as the analysis is the same for the second half):
    \begin{align*}
        \E[g_{i,j}'] &= \Pr[g_{i,j}' = 1] 
        = \sum_{k=0}^{d} \Pr\left[g_{i,j}' = 1 \middle|\sum_{j=1}^{d} g_{i,j} = k \right] \Pr\left[\sum_{j=1}^{d} g_{i,j} = k \right] \\
        &\leq \sum_{k=0}^{d} \frac{k}{d/2} \Pr\left[\sum_{j=1}^{d} g_{i,j} = k \right] 
        = \frac{2}{d} \E\left[\sum_{j=1}^{d} g_{i,j} \right]
        \leq \eta (1-c) .
    \end{align*}
    where the second line used that in the worst case where the $k$ disagreements all happen within the first half of the coordinates, after randomly permuting these coordinates, the probability of our fixed coordinate $j$ to experience a disagreement is $\tfrac{k}{d/2}$.
    This means that fixing a coordinate $j \in [d/2]$, the number of corruptions in that coordinate across all $n$ samples, $\sum_{i = 1}^n g_{i,j}'$, is a sum of independent binary variables with expectation $\eta(1-c)$ each. By Chernoff-Hoeffding bounds (\Cref{fact:Chernoff}), we obtain
    \begin{align*}
        \Pr\left[ \frac{1}{n}\sum_{i = 1}^n g_{i,j}' > \eta \right] \leq e^{-(1-c)\eta n \left(\frac{1}{1-c} - 1 \right)^2/3} = e^{- \Omega(c^2 n \eta)}.
    \end{align*}
    The same analysis applies for coordinates $j$ in the second half of the coordinates. By a union bound, the probability that there exists a coordinate $j \in [d]$ with $\frac{1}{n}\sum_{i = 1}^n g_{i,j}' > \eta$ is at most $d e^{- \Omega(c^2 n \eta)}$. This means that the event $\cE_1$ defined earlier has probability at least $1-d e^{- \Omega(c^2 n \eta)}$. Finally, we can define a similar  event $\cE_2$ for the label corruptions, i.e., $\cE_2$ being the even that $r_{d+1}$ does not reach zero. With another application of the Chernoff bound we can also conclude that $\cE'$ happens with probability at least  $1-e^{- \Omega(c^2 n \eta)}$. Combining with a union bound, the probability of both $\cE$ and $\cE_2$ happening is at most $1-(d+1)e^{- \Omega(c^2 n \eta)}$. This means that the distributions of the sets $S,S'$ output by the pseudocode have total variation distance at most $1-(d+1)e^{- \Omega(c^2 n \eta)}$. Consequently, by Le Cam's inequality \cite{lecam1973convergence}, any test that distinguishes between the two distributions has probability of failure at least $\frac{1}{2}(1-(d+1)e^{- \Omega(c^2 n \eta)})$.
\end{proof}

\LEMMAHYBRID*

\noindent Before presenting the proof, we note that although the statement is written in terms of Gaussian distributions (which is the setting we will apply it to later), the argument applies to any distributions $D$ and $D'$, where $D'$ is a shifted version of $D$ (in the sense that a random variable from $D'$ can be written as a random variable from $D$ plus a deterministic vector).

\begin{proof}

    This lemma essentially follows from the observation that two Gaussians with the same covariance and means differing in only one coordinate can be coupled so that the disagreement occurs only on that coordinate (and with probability at most equal to their total variation distance) while all other coordinates always agree. This is formalized in the claim below:
    
    \CLAIMONESTEP*

    \begin{proof}(Proof of \Cref{cl:one_step})
        Without loss of generality we use $i=1$ in this proof.
        Denote by $(X_1,\ldots,X_d)$ a random vector distributed as $Q=\cN((\mu_1,\mu_2,\ldots,\mu_d), \Sigma)$ and by $(X_1',\ldots,X_d')$ a random vector distributed as $Q'=\cN((\mu_1',\mu_2\ldots,\mu_{d}), \Sigma)$.
        Also, for any $Z \in \R^{d-1}$ denote by $P_Z$ the distribution of $X_1$ conditioned on $(X_2,\ldots,X_d) = Z$ and by $P'_Z$ the distribution of $X'_1$ conditioned on $(X'_2,\ldots,X'_d) = Z$.
        The coupling $\Pi$ that satisfies the guarantee in the claim statement is the distribution between the pair of vectors $(\tilde X_1,\ldots, \tilde X_d), (\tilde Y_1,\ldots, \tilde Y_d)$ created as follows:
        \begin{enumerate}
            \item Draw $Z \in \R^{d-1}$, according to the marginal distribution of $Q$ in the coordinates $2,3,\ldots,d$. (Note that this marginal is the same under $Q$ and $Q'$).
            \item Set $(\tilde X_2,\ldots, \tilde X_d) = Z$ and  $(\tilde Y_2,\ldots, \tilde Y_d) = Z$.
            \item Draw $\tilde X_1, \tilde Y_{1}$ from the maximal coupling $\Pi^*_Z$ (given in \Cref{fact:maximal}) between the distributions $P_Z$ and $P_{Z}'$.
        \end{enumerate}
        By construction, the marginal of $(\tilde X_1,\ldots, \tilde X_d)$ is $Q$ and that of $(\tilde Y_1,\ldots, \tilde Y_d)$ is $Q'$ thus $\Pi$ is a valid coupling. We also trivially have that $\Pr_{(X,X') \sim \Pi}[\tilde X_j \neq \tilde Y_j]=0$ for all $j \neq 1$. For the first coordinate, we have
        \begin{align*}
            \Pr_{(\tilde X, \tilde Y) \sim \Pi}[\tilde X_1 \neq \tilde Y_1 ] 
            &= \E_{Z} \left[ \Pr_{ (\tilde X_1, \tilde Y_{1})  \sim \Pi^*_Z}[X_1 \neq \tilde Y_1  \; | \; Z]  \right] \\
            &=  \E_{Z} \left[ \DTV\left( P_Z, P'_{Z}  \right) \right] \tag{using \Cref{fact:maximal}}\\
            &= \DTV\left( Q, Q'  \right) . \tag{law of total expectation}
        \end{align*}
    \end{proof}

    We now show how \Cref{lem:hybrid} follows given \Cref{cl:one_step}. We will show a procedure to generate random variables $X^{(i)} \in \R^d$ for $i = 0,\ldots,d$ and the coupling $\Pi$ that realizes \Cref{lem:hybrid} will be the joint distribution of $X^{(0)}$ and $X^{(d)}$. The generating procedure is the following:
    \begin{enumerate}
        \item Let $\Pi_1$ be the coupling that \Cref{cl:one_step} gives for the Gaussians $Q_0 =\cN((\mu_1,\mu_2,\ldots,\mu_d), \Sigma)$ and $Q_1 = \cN((\mu_1',\mu_2,\ldots,\mu_d), \Sigma)$. Draw $X^{(0)}$ and $X^{(1)}$ from $\Pi_1$. 
        \item Let $\Pi_2$ be the coupling that \Cref{cl:one_step} gives for the Gaussians $Q_1 = \cN((\mu_1',\mu_2,\mu_3,\ldots,\mu_d), \Sigma)$ and $Q_2 = \cN((\mu_1',\mu_2',\mu_3,\ldots,\mu_d), \Sigma)$. Draw $X^{(2)}$ from $\Pi_2$ conditioned on the value of $X^{(1)}$ from the previous step.
        \item In general, in the $i$-th step, let $\Pi_i$ be the coupling that \Cref{cl:one_step} gives for the Gaussians $Q_{i-1} = \cN((\mu'_1,\ldots,\mu'_{i-1},\mu_i,\ldots,\mu_d), \Sigma)$ and $Q_{i} = \cN((\mu'_1,\ldots,\mu'_{i-1},\mu'_i,\mu_{i+1},\ldots,\mu_d), \Sigma)$. Draw $X^{(i)}$ from $\Pi_i$ conditioned on the value of $X^{(i-1)}$ from the previous step.
    \end{enumerate}

\noindent The expected number of coordinates that disagree is bounded as follows:
    \begin{align*}
    \E\left[ \sum_{i=1}^d \1(X^{(0)}_i \neq X^{(d)}_i)  \right]
       &= \sum_{i=1}^d \E_{X^{(i-1)}} \left[ \E_{X^{(i)}}[ \1(X^{(i)}_i \neq X^{(i-1)}_i) \, | \, X^{(i-1)}] \right]\\
       &= \sum_{i=1}^d \E_{X^{(i-1)},X^{(i)}}[ \1(X^{(i)}_i \neq X^{(i-1)}_i)] \\
       &= \sum_{i=1}^d \DTV(Q_{i-1}, Q_i),
    \end{align*}
    where the first line follows from the fact that a disagreement between $X^{(0)}$ and $X^{(d)}$ in the $i$-th coordinate can occur only during the $i$-th step of the generation process (by the design of the couplings in \Cref{cl:one_step}). The transition from the second to the third line uses the law of total probability to rearrange the expectations; and the final line follows from the guarantee provided by the couplings in \Cref{cl:one_step}.

\end{proof}

\subsection{Proof of \Cref{thm:combined}\ref{it:thm-small-beta}}\label{app:proof-of-small-beta} 

We restate and prove the following lower bound.

\begin{restatable}[Lower Bound for regime $\|\beta\| \leq \sigma$]{theorem}{SMALLBETA}\label{thm:small-beta}
    Let $c$ be a sufficiently small positive absolute constant. For any $d \in \Z_+$, $\sigma,\eta,b \in \R_+$ with $\eta \in [0,1]$, $\sigma > 0$ and  $b \leq \sigma$ the following statement holds.
    For every algorithm $\cA$ that takes as input  $\eta,\sigma,b$ as well as $n$ labeled examples $\{ (x^{(i)},y^{(i)}) \}_{i=1}^n$ with $x^{(i)}$ and $y^{(i)} \in \R$ and outputs a vector $\hat{\beta}\in \R^d$, there exists a $\beta \in \R^d$ with $\|\beta\|=b$ such that running $\cA$ on input $\eta,\sigma,b$ and $n$ labeled examples from the model of \Cref{def:linear_model} with regressor $\beta$,  standard deviation $\sigma$ for the additive noise, and $\eta$-fraction of missing data per coordinate according to the contamination model of \Cref{def:model}, the output $\hat{\beta}$ satisfies
    \begin{align}
        \left\| \hat{\beta} - \beta \right\| \geq c \, \min(\|\beta\|,\eta \sqrt{d} \sigma). 
    \end{align}
    with probability at least $\frac{1}{2}(1-(d+1) e^{- \Omega(\eta n)})$.
\end{restatable}

\begin{proof}
We first focus on proving the result for $\|\beta\| \leq \eta \sqrt{d} \sigma$,  in which case $\min(\|\beta\|,\eta \sqrt{d} \sigma) = \|\beta\|$. We describe at the end how to extend it for $\|\beta\| > \eta \sqrt{d} \sigma$.    

For any $b \in  [0, \eta \sqrt{d} \sigma]$ we define the hypothesis testing problem of distinguishing between the regression vectors $\beta^{(0)},\beta^{(1)}$ defined below from $n$ samples from the models of \Cref{def:linear_model,def:model}:
\begin{enumerate}
    \item (Null Hypothesis) $\beta^{(0)} = (b/\sqrt{d},\ldots,b/\sqrt{d})$.
    \item (Alternative Hypothesis) $\beta^{(1)} = - \beta^{(0)}$.
\end{enumerate}
Note that $\|\beta^{(0)}\| = \| \beta^{(1)} \| =  b$ and $\|\beta^{(0)} - \beta^{(1)}\| = \sqrt{2}  b$.
By the standard reduction from estimation to hypothesis testing, showing that no algorithm can solve the above testing problem will imply that no estimator has Euclidean error smaller than $b/\sqrt{2}$.
By \Cref{fact:conditional}, for each of the two hypotheses $i \in \{0,1\}$ above, the conditional distribution $X|(y=t)$ of the covariates given that the label is $t$ is a Gaussian $\cN(\mu^{(i)},\Sigma^{(i)})$ where 
    \begin{align}
        \mu^{(i)} := \frac{t}{\sigma^2+ b^2}\beta^{(i)}, 
        \quad \text{and} \quad \Sigma^{(1)} = \Sigma^{(2)} = \Sigma := I - \frac{\beta^{(1)}(\beta^{(1)})^\top}{\sigma^2 + b^2}. \label{eq:mean-cov}
    \end{align}
\noindent 
In order to prove that the testing problem is not solvable, it suffices to find a coupling between $\cN(\mu^{(0)},\Sigma),\cN(\mu^{(1)},\Sigma)$ such that for each $i \in [d]$, only up to $\eta$-fraction of samples have their $i$-th coordinate corrupted. 
By \Cref{lem:hybrid}, for each $i \in [d]$ we need to upper bound the TV-distance between the two Gaussians $\cN(m^{(i)}, \Sigma), \cN(m^{(i+1)} , \Sigma)$, where
\begin{align}
m^{(i)} = (\mu^{(1)}_1, \ldots, \mu^{(1)}_{i-1}, \mu^{(0)}_{i}, \mu^{(0)}_{i+1}, \ldots, \mu^{(0)}_{d}), \label{eq:interm1-first}\\
m^{(i+1)} = (\mu^{(1)}_1, \ldots, \mu^{(1)}_{i-1}, \mu^{(1)}_{i}, \mu^{(0)}_{i+1}, \ldots, \mu^{(0)}_{d}) . \label{eq:interm2-first} 
\end{align}

We will do this by bounding the KL-divergence and relating it to the total variation distance via Pinsker's inequality (\Cref{fact:pinsker}).
By \Cref{fact:KL-Gaussians}, we have
\begin{align}
    \DKL&(\cN(m^{(i)} , \Sigma)  \,\|\, \cN(m^{(i+1)} , \Sigma)) \label{eq:firsteq}\\
    &\leq \frac{1}{2}(m^{(i+1)} - m^{(i)})^\top \Sigma^{-1} (m^{(i+1)} - m^{(i)}) \\
    &\leq \frac{1}{2}(m^{(i+1)} - m^{(i)})^\top  \left( I + \frac{{\beta^{(1)}} {\beta^{(1)}}^\top}{\sigma^2} \right) (m^{(i+1)} - m^{(i)}) \tag{using  \Cref{cl:sherman-morrison}}\\
    &\leq \frac{1}{2} \| m^{(i+1)} - m^{(i)} \|^2 + \frac{1}{2}\frac{((m^{(i+1)} - m^{(i)})^\top {\beta^{(1)}})^2}{\sigma^2} \\
    &\leq  \frac{2t^2}{d} \frac{b^2}{(\sigma^2 + b^2)^2} + \frac{1}{2\sigma^2}\left( \frac{2t}{\sqrt{d}} \frac{b}{\sigma^2 + b^2} \frac{b}{\sqrt{d}}\right)^2 \tag{see explanation below}\\
    &= \frac{2t^2}{d} \frac{b^2}{(\sigma^2 + b^2)^2} + \frac{1}{\sigma^2}\frac{2t^2}{d} \frac{b^2}{(\sigma^2 + b^2)^2} \frac{b^2}{d}\\
    &= \frac{2t^2}{d} \frac{b^2}{(\sigma^2 + b^2)^2}\left( 1 + \frac{b^2}{d \sigma^2} \right)\\
    &\leq\frac{2t^2}{d} \frac{b^2}{(\sigma^2 + b^2)^2}(1 + \eta^2) \tag{using assumption $b \leq \eta \sqrt{d} \sigma$}\\
    &\leq \frac{4t^2}{d} \frac{b^2}{(\sigma^2 + b^2)^2},  \tag{using $\eta \leq 1$}
\end{align}
where the fourth inequality above follows from equations (\ref{eq:mean-cov}) to (\ref{eq:interm2-first}) as follows: the vector $m^{(i+1)} {-} m^{(i)}$ has zero in all coordinates except the $i$-th, where it equals $\frac{t}{\sigma^2 + b^2} \cdot \frac{b}{\sqrt{d}}$.

Combining the above with Pinsker's inequality (\Cref{fact:pinsker}), we obtain 
\begin{align}\label{eq:lasteq}
    \DTV(\cN(m^{(i)} , \Sigma), \cN(m^{(i+1)} , \Sigma)) 
    \leq \frac{\sqrt{2} |t|}{\sqrt{d}}  \frac{b}{\sigma^2 + b^2} .
\end{align}

By \Cref{lem:hybrid}, there exists a coupling $\Pi$ between $\cN(m^{(0)},\Sigma),\cN(m^{(d)},\Sigma)$ such that
\begin{align}\label{eq:coupling_dissagreements}
    \E_{(X,X') \sim \Pi}\left[ \sum_{i=1}^d \1( X_i \neq X_i') \right] \leq \frac{\sqrt{2} \sqrt{d} |t|  b}{\sigma^2 + b^2}.
\end{align}

As explained earlier, the Gaussians involved in the above coupling are the conditional distributions of the covariates in our linear regression model (\Cref{def:linear_model}), given that the label $y$ equals $t$. We can easily extend this coupling to a new coupling $\Pi'$, which couples the entire labeled examples from our linear regression model under the two hypotheses, as follows (recall that the distribution of labels under both hypotheses is $\cN(0,\sigma^2 + \|\beta^{(i)}\|^2) = \cN(0,\sigma^2 + b^2)$:

\begin{enumerate}
    \item Draw $t \sim \cN(0,\sigma^2 + b^2)$ and set $y = y' = t$.
    \item Draw $(X,X')$ from the coupling $\Pi$ that satisfies \Cref{eq:coupling_dissagreements}.
    \item Return $(X,y),(X',y)$.
\end{enumerate}
By construction, $(X,y)$ is marginally distributed according to the linear regression model with regressor $\beta^{(0)}$ and $(X',y')$ is marginally distributed according to the linear regression model with regressor $\beta^{(1)}$, thus $\Pi'$ is a valid coupling. Moreover, the expected number of disagreements is
\begin{align*}
     \E_{((X,y),(X',y') \sim \Pi'}\left[ \sum_{i=1}^d \1(X_i \neq X_i') \right] 
     &\leq  \frac{\sqrt{2}\sqrt{d}\,  b\E_{t \sim \cN(\sigma^2 + b^2)}[|t|]}{\sigma^2 + b^2}
     \leq  \frac{\sqrt{2} \sqrt{d}\,  b}{\sqrt{\sigma^2 + b^2}} \tag{using \Cref{eq:coupling_dissagreements}}\\
     &\leq  \frac{\sqrt{2} \sqrt{d}\, b}{\sigma} \tag{using $b \geq 0$}\\
     &=    \frac{\sqrt{2} \sqrt{d}\eta \sigma \sqrt{d} }{\sigma} \tag{using $b \leq \sqrt{d}\eta \sigma $}\\
     &\leq \sqrt{2}  \eta  d \leq \sqrt{2}  \eta d.
\end{align*}
The labels always agree in this coupling. The constant $\sqrt{2}$ in front of $\eta$ above is not particularly important; we could obtain the same bound with a smaller constant, such as $1/3$, in front of $d$ by simply redefining $\eta$ as $\eta / (2\sqrt{2})$ at the beginning of the proof. This would allow us to apply \Cref{fact:coupling} to conclude that no algorithm can solve the hypothesis testing problem (except with probability $\frac{1}{2}(1-(d+1) e^{-\Omega(\eta n)})$), thus completing the proof of \Cref{thm:small-beta} for the regime $\|\beta\| \leq \eta \sqrt{d} \sigma$.

For the remaining regime $\eta \sqrt{d} \sigma \leq \|\beta\|$, consider instead the hypothesis testing problem over $\R^{d}$: distinguish $\beta = (r, b/\sqrt{d}, \ldots, b/\sqrt{d})$ from $\beta' = (r, -b/\sqrt{d}, \ldots, -b/\sqrt{d})$, where $b \leq \eta \sqrt{d}\sigma$ is as before and $r$ is a tunable parameter allowing $\|\beta\|$ to get any desired value larger than $\eta \sqrt{d} \sigma$. The labels can be written as $y = r X_1 + y_0$ and $y' = r X_1' + y_0'$, where $y_0 = (b/\sqrt{d}, \ldots, b/\sqrt{d})^\top X_{2:d} + \xi$ and $y_0' = (-b/\sqrt{d}, \ldots, -b/\sqrt{d})^\top X_{2:d}' + \xi'$, with $X_{2:d}$ denoting the last $d$ coordinates of $X$. As shown earlier, there exists a coupling between $(X_{2:d}, y_0)$ and $(X_{2:d}', y_0')$ with expected disagreements at most $\eta d/2$ per sample. This extends trivially to a coupling between $(X, y)$ and $(X', y')$ (by adding shared Gaussian noise $\mathcal{N}(0, r^2)$ to the labels), preserving the disagreement bound. Applying \Cref{fact:coupling} as before completes the proof.
\end{proof}

\section{Omitted Details from \Cref{sec:thm-proofs}}\label{app:thm-proofs}

In this section we restate and prove the lower bounds corresponding to parts \ref{it:thm-big-eta},\ref{it:thm-interm-eta} and \ref{it:thm-small-eta} of \Cref{thm:combined}.
Since these bounds depend on \Cref{lem:impr-coupling}, we restate that lemma below for convenience:

\IMPRCOUPLING*

We start with \Cref{thm:combined}\ref{it:thm-big-eta}

\begin{restatable}[Lower bound for regime $\eta \geq 7/\sqrt{d}$]{theorem}{IMPROVEDLBOUNDTWO}\label{thm:big-eta}
    The following holds for every $d \in \Z_+$ and $\sigma,\eta, b \in \R_+$ with  $7/\sqrt{d} \leq \eta \leq 1$.
    For every algorithm $\cA$ that takes as input  $\eta,\sigma,b$ as well as $n$ labeled examples $\{ (x^{(i)},y^{(i)}) \}_{i=1}^n$ with $x^{(i)} \in \R^d$ and $y^{(i)} \in \R$ and outputs a vector $\hat{\beta} \in \R^d$, there exists a $\beta \in \R^d$ with $\|\beta\|=b$ such that running $\cA$ on input $\eta,\sigma,b$ and $n$ labeled examples from the model of \Cref{def:linear_model} with regressor $\beta$,  standard deviation $\sigma$ for the additive noise, and $\eta$-fraction of missing data per coordinate according to the contamination model of \Cref{def:model}, the output $\hat{\beta}$ satisfies
    \begin{align}
        \left\| \hat{\beta} - \beta \right\| \geq \frac{1}{\sqrt{2}}   \|\beta\| , 
    \end{align}
    with probability at least $\frac{1}{2}(1-(d+1) e^{-\Omega(\eta n)})$.
\end{restatable}

\begin{proof}
    Let $\beta = s(1,1\ldots,1)$ be the all-ones vector and $\beta' = -\beta$ the vector with $-1$ in every coordinate where $s$ is a tunable parameter so that $\|\beta\|$ can have any desired value (we want to prove that the lower bound holds for any $\|\beta\|$). We consider the following hypothesis testing problem:
    \begin{itemize}
        \item (Null Hypothesis) The regression vector is  $\beta$.
        \item (Alternative Hypothesis) The regression vector is  $\beta'$.
    \end{itemize}
    It suffices to show that no algorithm distinguishes between the two hypotheses with probability better than $\frac{1}{2}(1+(d+1) e^{- \Omega(\eta n)})$. Since $\|\beta - \beta'\| = \sqrt{2d} \geq \sqrt{2} \max(\|\beta\|,\|\beta'\|)$, by the standard reduction between estimation and hypothesis testing, this would imply that no algorithm can estimate $\beta$ with error smaller than $\|\beta\|/\sqrt{2}$.

    First, it is easy to see that the coupling construction from \Cref{lem:impr-coupling} allows us to couple the distributions of labeled examples for the two hypotheses while ensuring a small number of coordinate-wise disagreements. This is shown in the lemma below:

\begin{lemma}\label{lem:coupling_whole_thing_simple}
Let $\sigma \geq 0$ and a scaling parameter $s \geq 0$.
Define $\beta = s (1,1,\ldots,1)$ as a scaled version of the all-ones vector in $\R^d$ and $\beta' = -\beta$.
    If $D$ denotes the distribution of a labeled example $(X,y)$ drawn from the linear regression model of \Cref{def:linear_model} with regressor $\beta$ and standard deviation of additive noise $\sigma$, and $D'$ denotes the distribution of a labeled example according to a linear regression model with regressor $\beta'$ and standard deviation of additive noise $\sigma$, then there exists a coupling $\Pi$ between $D$ and $D'$ such that $\Pr_{((X,y),(X',y')) \sim \Pi}[y = y'] = 1$ and
    \begin{align}\label{eq:coupling_guarrantee2}
        \E_{((X,y),(X',y')) \sim \Pi}\left[ \sum_{i=1}^d \1\left(  X_i \neq X_i'\right) \right] \leq 3\sqrt{d} .
    \end{align}
\end{lemma}
\begin{proof}(Proof of \Cref{lem:coupling_whole_thing_simple})
    It suffices to prove the lemma for $s=1$ and $\sigma = 0$. This is because if $((X,y),(X',y'))$ is distributed according to a coupling with the desired properties ($\Pr[y = y'] = 1$ and $\E\left[ \sum_{i=1}^d \1\left(  X_i \neq X_i'\right) \right] \leq 3\sqrt{d}$) for the case of $s=1$ and $\sigma = 0$, then we can let $\xi \sim \cN(0,\sigma^2)$ and the pair $((sX,sy + \xi),(sX',sy' + \xi))$ will be the final coupling that corresponds to the case with positive $\sigma$ and $s \neq 1$ and continues to satisfy $\Pr[y = y'] = 1$ and $\E\left[ \sum_{i=1}^d \1\left(  X_i \neq X_i'\right) \right] \leq 3\sqrt{d}$. The coupling for the case $s=1,\sigma = 0$ is the following:
    \begin{itemize}
        \item Draw $z \sim \cN(0,d)$ and set $y = y' = z$.
        \item Draw $((X_1,\ldots,X_d),(X_1',\ldots,X_d'))$ from the coupling of \Cref{lem:impr-coupling} applied with $t=z$ and $t' = -z$.
    \end{itemize}

    It is easy to verify that the labeled example $(X_1,\ldots,X_d,y)$ is marginally distributed according to the linear model that uses $\beta = (1,\ldots,1)$ as the regressor, and the example $(X_1',\ldots,X_d',y)$ is distributed according to the linear model using $\beta' = (-1,\ldots,-1)$. Moreover, by \Cref{lem:impr-coupling}, the expected number of disagreements is 
    \begin{align*}
        \E_{(X_1,\ldots,X_d),(X_1',\ldots,X_d')}\left[ \1(X_i \neq X_i') \right] \leq 1 + \E_{t,t'}[|t-t'|]  = 1 + \E_{z \sim \cN(0,d)}[|2z|]\leq 3 \sqrt{d}.
    \end{align*}
    
    The proof of \Cref{lem:coupling_whole_thing_simple} is completed.
\end{proof}

    When $\eta \geq 7/\sqrt{d}$ the right hand side of \Cref{eq:coupling_guarrantee2} is at most $\eta d 3/7$. Thus, by \Cref{fact:coupling} , no algorithm can solve the hypothesis testing defined in the beginning with probability higher than $\frac{1}{2}(1+(d+1)e^{-\Omega(\eta n)})$.

\end{proof}

We now move to \Cref{thm:combined}\ref{it:thm-interm-eta} which is restated and proved below.

\begin{restatable}[Lower bound for regime $\frac{2+c}{d} \leq \eta \leq \frac{7}{\sqrt{d}}$]{theorem}{IMPROVEDLBOUND}\label{thm:interm-eta}
There exists a sufficiently large absolute constant $C$ such that the following holds for every $d \in \Z_+$, every $c \in (0,1)$, every $\sigma \geq 0$, every $\eta \in [\frac{2+c}{d}, \frac{7}{\sqrt{d}}]$, and every $b \in \R_+$. For every algorithm $\cA$ that takes as input  $\eta,\sigma,b$ as well as $n$ labeled examples $\{ (x^{(i)},y^{(i)}) \}_{i=1}^n$ with $x^{(i)} \in \R^d$ and $y^{(i)} \in \R$ and outputs a vector $\hat{\beta} \in \R^d$, there exists a $\beta \in \R^d$ with $\|\beta\|=b$ such that running $\cA$ on input $\eta,\sigma,b$ and $n$ labeled examples from the model of \Cref{def:linear_model} with regressor $\beta$,  standard deviation $\sigma$ for the additive noise, and $\eta$-fraction of missing data per coordinate according to the contamination model of \Cref{def:model}, the output $\hat{\beta}$ satisfies
    \begin{align}
        \left\| \hat{\beta} - \beta \right\| \geq   \eta \frac{c\sqrt{d}}{C} \|\beta\|  ,  
    \end{align}
    with probability at least $\frac{1}{2}(1-(d + 1) e^{- \Omega(c^2\eta n)})$.
\end{restatable}
\begin{proof}

    Let $s \geq 0$ be a scaling factor to control the norm of the regressor. 
    We will use a parameter $\eps \in (0,1)$ that will be specified later.
    Fix $\beta = s (\eps, \dots, \eps, 1, \dots, 1)$ and $\beta' = s (-\eps, \dots, -\eps, 1, \dots, 1)$ as two vectors in $\R^d$, where the first $d/2$ coordinates are $\eps s$ for $\beta$ and $-\eps s$ for $\beta'$, while the last $d/2$ coordinates are all equal to $s$.
    By the estimation to hypothesis testing reduction discussed in \Cref{sec:hybrid}, it suffices to show that no algorithm that uses $n$ $\eta$-corrupted samples can solve  the following hypothesis testing problem:
    \begin{itemize}
        \item (Null Hypothesis) The regression vector is  $\beta$.
        \item (Alternative Hypothesis) The regression vector is  $\beta'$.
    \end{itemize}

    Note that we can use $s$ to make the norm of $\|\beta\|$ and $\|\beta'\|$ have any desired value $b$ (as in the statement of \Cref{thm:interm-eta}).
    Thus we focus on showing hardness of the testing problem in what follows.

    Using the coupling construction from \Cref{sec:improved-coupling}, we can construct a coupling between the linear regression distributions corresponding to the two hypotheses such that the number of disagreements per coordinate is upper bounded as stated in \Cref{lem:coupling_whole_thing} below. 

    \begin{lemma}\label{lem:coupling_whole_thing}
Let $d$ be a power of $2$, $\eps \in [0,1]$, $\sigma \geq 0$, and $s \geq 0$.
Let $\beta = s(\eps,\ldots,\eps,1,\ldots,1)$ and $\beta' = s(-\eps,\ldots,-\eps,1,\ldots,1)$ be the two vectors in $\R^d$, where the first $d/2$ coordinates are $s\, \eps$ for $\beta$ and $-s\,\eps$ for $\beta'$, while the last $d/2$ coordinates are all equal to $s$.
    If $D$ denotes the distribution of a labeled example drawn from the linear model of \Cref{def:linear_model} with regressor $\beta$ and standard deviation of additive noise $\sigma$, and $D'$ denotes the distribution of a labeled example according to a linear model with regressor $\beta'$ and standard deviation of additive noise $\sigma$, then there exists a coupling $\Pi$ between $D$ and $D'$ such that $\Pr_{((X,y),(X',y')) \sim \Pi}[y = y'] = 1$ and
    \begin{align}\label{eq:coupling_guarrantee}
        \E_{((X,y),(X',y')) \sim \Pi}\left[ \sum_{i=1}^d \1\left(  X_i \neq X_i'\right) \right] = 2 + O\left( \eps \sqrt{d}  \right).
    \end{align}
\end{lemma}

The proof of \Cref{lem:coupling_whole_thing} is somewhat tedious, so we defer it until after we show how the theorem follows from \Cref{lem:coupling_whole_thing}. 
Denote by $C$ the absolute constant hidden inside the big-O notation in \Cref{eq:coupling_guarrantee}, and denote by $c'$ an additional parameter that is less than $0.25$.
We will re-parameterize things as follows:  $\eta = \frac{2}{(1-c')d} + \frac{C \eps}{(1-c')\sqrt{d}}$, or equivalently $\eps := (\eta - \frac{2}{(1-c')d})\frac{\sqrt{d}}{C }$. This is so that the right-hand side of \Cref{eq:coupling_guarrantee} becomes equal to $\eta d (1-c')$ and we can use \Cref{fact:coupling} to conclude that no algorithm can solve the hypothesis testing problem with probability better than $\frac{1}{2}(1+(d+1)e^{-\Omega(c' \eta n)})$. By the estimation to hypothesis testing reduction, this means that every estimation algorithm has error at least $\Omega(\|\beta - \beta'\|)$. For the precise form of this error we have the following lower bounds:
    \begin{align*}
        \|\beta - \beta'\| &= s\, \eps \sqrt{d/2} \tag{by definition of $\beta,\beta'$}\\
    &= \frac{\eps}{\sqrt{1+\eps^2}} \max(\|\beta\|, \|\beta'\|) \tag{$\|\beta\| = \|\beta'\| = s\sqrt{d(1 + \eps^2)/2}$}\\
    &\gtrsim \eps \max(\|\beta\|, \|\beta'\|) \tag{$\eps := (\eta - \frac{2}{(1-c')d})\tfrac{\sqrt{d}}{C} \leq 1$ since $\eta \leq 7/\sqrt{d}$}\\
    &\gtrsim \left( \eta - \frac{2}{(1-c')d} \right) \sqrt{d} \max\left(\|\beta\|, \|\beta'\|\right) \tag{using $\eps := (\eta - \frac{2}{(1-c')d})\tfrac{\sqrt{d}}{C}$}\\
    &\gtrsim  c' \eta \sqrt{d} \max\left(\|\beta\|, \|\beta'\|\right)   , \tag{this step holds if $\eta \geq \frac{2+10c'}{d}$.} 
    \end{align*}
    where the last step uses that $\eta \geq \frac{2+10c'}{d}$ implies $\eta \geq \frac{2+10c'}{d} \geq \frac{2}{(1-c')^2d}$ and means that the terms in the parentheses from the previous step can be lower bounded as $\eta - \frac{2}{(1-c')d} \geq c'\eta$. The fact that $\eta \geq \frac{2+10c'}{d}$ is due to the assumption in the theorem statement that $\eta \geq (2 + c)/d$ (and the fact that $c'$ is a parameter that we can choose to be $c' = c/10$).

We now prove \Cref{lem:coupling_whole_thing} that was used earlier.

\begin{proof}(Proof of \Cref{lem:coupling_whole_thing})
    We first note that it suffices to prove the claim for $s=1$ and $\sigma = 0$. To see that, let us momentarily use the notation $D_{s,\sigma}, D_{s,\sigma}'$ for the two distributions of the statement explicitly indicating the values of $s$ and $\sigma$. Suppose that there exists the desired coupling $\Pi$ between $D_{1,0}$ and $D'_{1,0}$ (i.e., for $s=1$ and $\sigma = 0$). If  $((X,y),(X',y')) \sim \Pi$ then $((sX,sy),(sX',sy'))$ is a coupling with the desired properties for $D_{s,0}, D_{s,0}'$, i.e., $\Pr[sy = sy']=1$ and $\E[\sum_{i=1}^d \1\left(  X_i \neq X_i'\right) ] = 2 + O(\eps \sqrt{d})$. If we further consider $\xi \sim \cN(0,\sigma^2)$, then the pair $((sX,sy + \xi),(sX',sy' + \xi))$ is the desired coupling between $D_{s,\sigma}, D_{s,\sigma}'$.

    For the remainder of the proof, we thus use $s = 1$ and $\sigma = 0$. We now define how the coupling that generates the pair $((X, y), (X', y'))$. While the definition may initially appear complicated, we will argue below that the resulting distribution is indeed a valid coupling and ultimately show that it satisfies $y = y'$ as well as \Cref{eq:coupling_guarrantee}.
    \begin{enumerate}
        \item We first sample $z \sim \cN(0, (1+\eps^2)d/2)$ and set $y=y'=z$. \label{it:coupling1}
        \item We then sample $t$ and $t'$ from a coupling such that the marginals are $t \sim \cN(\eps z/(1+ \eps^2), d/(2(1+\eps^2))$ and $t' \sim \cN(-\eps z/(1+ \eps^2), d/(2(1+\eps^2))$ and it also holds $t-t' = 2\eps z/(1+ \eps^2)$ always (i.e., with probability $1$).
        \item We sample $(X_1,\ldots,X_{d/2}), (X'_1,\ldots, X'_{d/2})$ according to the coupling from \Cref{lem:impr-coupling}, i.e., the distribution of $(Z_1,\ldots,Z_{d/2}) \sim \cN(0,I)$ conditioned on $\sum_{i \in [d/2]} Z_i = t$ and the distribution of $(Z'_1,\ldots, Z'_{d/2}) \sim \cN(0,I)$ conditioned on $\sum_{i \in [d/2]} Z'_i = t'$.
        \item We sample $(X_{d/2},\ldots,X_d),(X'_{d/2},\ldots,X'_d)$ from the coupling of \Cref{lem:impr-coupling} but now using the conditioning $\sum_{i = d/2+1}^d Z_i = z - \eps t$ and $\sum_{i = d/2+1}^d Z_i' = z + \eps t'$ respectively.
    \end{enumerate}
    
    We start by verifying that this is indeed a valid coupling, i.e., that $(X_1,\ldots,X_d,y)$ follows the linear model with regressor $\beta$ and $(X'_1,\ldots,X'_d,y')$ the linear model with regressor $\beta'$. We will perform the check for the first part (i.e., check that $(X_1,\ldots,X_d,y)$ indeed follows the linear model with regressor $\beta$). Checking the other part can be done with the same argument and replacing $\eps$ with $-\eps$. For notational purposes only, we let another example $(\tilde X_1,\ldots\tilde X_d,\tilde y)$ be a random labeled example distributed according to the distribution $D$ from the lemma statement.

    First, the marginal distribution of  $\tilde y$ is $\tilde y = \beta^\top \tilde X \sim \cN(0,\|\beta\|^2) = \cN(0, (1+\eps^2)d/2)$. This explains the first step (\Cref{it:coupling1} in our coupling procedure).
    To explain the rest of the steps, we want to argue that $(X_1,\ldots, X_d) | (y = z)$ in our procedure is distributed in the same way as $(\tilde X_1,\ldots,\tilde X_d) | (\tilde y = z)$ (i.e, in the same way as under the distribution $D$ from the lemma statement). The distribution under $D$ can be viewed as follows: First we sample a value $t$ from the distribution of $\sum_{i \in [d/2]}\tilde X_i$ conditioned on $\tilde y=z$ (let us call this distribution $D_1$ for later referencing), then we sample the coordinates $(\tilde X_1,\ldots,\tilde X_{d/2})$ from the distribution of a standard Gaussian vector, conditioned on $\sum_{i \in [d/2]}\tilde X_i = t$ and finally we sample the rest of the coordinates $(\tilde X_{d/2},\ldots,\tilde X_{d})$ from the distribution of a standard Gaussian vector, conditioned on $\sum_{i =d/2+1}^d \tilde X_i = z - \eps t$ (so that $\tilde y = \beta^\top \tilde X = \sum_{i=1}^{d/2} \eps \tilde X_i + \sum_{i=d/2+1}^d \tilde X_i = \eps t + z - \eps t = z$). It can be checked that the distribution $D_1$ mentioned earlier is $\cN(\eps z/(1+\eps^2), d/(2(1+\eps^2))$. This can be seen as follows: First, by \Cref{fact:conditional} applied with $u=\beta$ and $\sigma=0$, the conditional distribution of the entire vector, given $y=z$ is
    \begin{align*}
          \tilde X_1,\ldots,\tilde X_{d/2},\ldots, \tilde X_d | (\tilde y =z)    \sim \cN\left( z \frac{\beta}{\|\beta\|^2}, I_d - \frac{\beta \beta^\top}{\|\beta\|^2} \right).
    \end{align*}
    This means that 
    \begin{align*}
        \tilde X_1,\ldots,\tilde X_{d/2} | \tilde y =z \sim \cN\left( z \frac{\beta_{1:d/2}}{\|\beta\|^2}, I_{d/2} - \frac{\beta_{1:d/2} \beta_{1:d/2}^\top}{\|\beta\|^2} \right),
    \end{align*}
    where the notation $\beta_{1:d/2}$ denotes the vector formed by taking the first $d/2$ coordinates of $\beta$ and $I_{d/2}$ is the $(d/2)\times(d/2)$ identity matrix. Note that $\beta_{1:d/2}$ is the vector in $\R^{d/2}$ with value $\eps$ in every coordinate. Then, by \Cref{fact:linear} applied with $u$ being the all-ones vector, we have
    \begin{align*}
        \sum_{i=1}^{d/2} \tilde X_i \big| (\tilde y=z) \sim \cN\left( \frac{\eps z}{(1+\eps^2)}, \frac{d}{2} \frac{1}{(1+\eps^2)} \right).
    \end{align*}
    This completes the proof that the procedure generating $\tilde X_1,\ldots,\tilde X_d,y$ that is described in the bullets at the start of this proof matches the distribution of a labeled example under the linear model with regressor $\beta$.
    The check for the marginal of $X'_1,\ldots,X'_d$ is similar and we skip it. So far we have thus shown that the generating process from the bullets defines a valid coupling $\Pi$ between $D$ and $D'$.

    Finally, we analyze the expected number of disagreements to show \Cref{eq:coupling_guarrantee}.
    Because of our two applications of \Cref{lem:impr-coupling}, this expected number of disagreements is
    \begin{align}
        \E_{((X,y),(X',y')) \sim \Pi}\left[ \sum_{i=1}^d \1\left(  X_i \neq X_i'\right) \right]
        &= \E_{t,t'}\left[ 2 + O\left(  |t - t'|   +  |z - \eps t - (z + \eps t')|  \right)  \right] \notag \\
        &= 2 + O\left(\E_{t,t'}\left[|t - t'| \right]  \right) + \eps \, O\left(\E_{t,t'}\left[|t + t'| \right]   \right). \label{eq:three_terms}
    \end{align}
    the second term above is $O( \eps \E[|z|]  /(1+\eps^2))$ because we have defined the coupling between $t$ and $t'$ to always satisfy $t-t' = 2\eps z/(1+\eps^2)$. The term can be further bounded by $O(\eps \sqrt{d}  )$ using that $z \sim \cN(0, (1+\eps^2)d/2)$ and $\eps \leq 1$. Regarding the last term in \Cref{eq:three_terms}, we have $$ \eps   \E_{t,t'}\left[|t + t'| \right]   \lesssim   \eps   \E_{t,t'}\left[|t | \right] +  \eps   \E_{t,t'}\left[| t'| \right] \lesssim \eps    \sqrt{d},$$ where we used the triangle inequality and the fact that the variance of $t$ and $t'$ is at most $d/2$. 
    This concludes the proof of \Cref{eq:coupling_guarrantee} and \Cref{lem:coupling_whole_thing}.
\end{proof}
    Since the proof of \Cref{lem:coupling_whole_thing} is complete, this also completes the proof of \Cref{thm:interm-eta}.
\end{proof}

We restate and prove \Cref{thm:combined}\ref{it:thm-small-eta} below.

\begin{restatable}[Lower bound for regime $\eta \leq 1/d$]{theorem}{REMAINING}\label{thm:small-eta}
There exists a sufficiently large absolute constant $C$ such that the following holds for every $d \in \Z_+$, $\sigma \geq 0$,  $\eta \in [0,1/d]$, and  $b \geq0$.
    For every algorithm $\cA$ that takes as input  $\eta,\sigma,b$ as well as $n$ labeled examples $\{ (x^{(i)},y^{(i)}) \}_{i=1}^n$ with $x^{(i)} \in \R^d$ and $y^{(i)} \in \R$ and outputs a vector $\hat{\beta} \in \R^d$, there exists a $\beta \in \R^d$ with $\|\beta\|=b$ such that running $\cA$ on $\eta,\sigma,b$ and $n$ labeled examples from the model of \Cref{def:linear_model} with regressor $\beta$,  standard deviation $\sigma$ for the additive noise, and $\eta$-fraction of missing data per coordinate according to the contamination model of \Cref{def:model}, the output $\hat{\beta}$ satisfies
    \begin{align}
        \left\| \hat{\beta} - \beta \right\| \geq   \frac{1}{C } \min\left(\eta d \sigma , \, \eta \sqrt{d} \|\beta\|    \right)  ,  
    \end{align}
    with probability at least $\frac{1}{2}(1-(d + 1) e^{- \Omega(\eta n)})$.
\end{restatable}

\begin{proof}
We let $B \in \R_+, E \in (0,1)$ be parameters where  $E$ will be specified later on as a function of $B$ and $\eta,d, \sigma$, and $B$ is a parameter that is chosen to ensure that the norm of $\beta$ and $\beta'$ below is equal to $b$ (recall that as stated in \Cref{thm:small-eta}, we want our lower bound to hold for any value of $\|\beta\|$).
We will consider the following hypothesis testing:
\begin{itemize}
    \item (Null Hypothesis) The regressor is $\beta := \left( \frac{B}{\sqrt{d/2}}\1_{d/2}, \frac{E}{\sqrt{d/2}}\1_{d/2} \right)$.
    \item (Alternative Hypothesis) The regressor is $\beta' := \left( \frac{B}{\sqrt{d/2}}\1_{d/2}, -\frac{E}{\sqrt{d/2}}\1_{d/2} \right)$.
\end{itemize}
To clarify the notation, $\1_{d/2}$ is the vector in $\R^{d/2}$ having the value $1$ in all its coordinates, thus $\beta$ above is the vector having $B/\sqrt{d/2}$ in the first $d/2$ coordinates and $E/\sqrt{d/2}$ in the second $d/2$ coordinates and $\beta'$ is defined similarly, but with $-E$ instead of $E$.

By the standard reduction from estimation to hypothesis testing (described at the beginning of \Cref{sec:hybrid}), showing that the above hypothesis testing problem is unsolvable implies an estimation error lower bound of $\Omega(\|\beta - \beta'\|) = \Omega(E)$.

In order to prove hardness of the hypothesis testing problem we prove the following lemma:
\begin{lemma}\label{lem:core-coupling-remaining}
    Let $B \in \R_+$ and $E \in (0,1)$ with $E \lesssim B$.
    Let $d \in \Z_+$  $\sigma \geq 0$ $\beta := \left( \frac{B}{\sqrt{d/2}}\1_{d/2}, \frac{E}{\sqrt{d/2}}\1_{d/2} \right)$ and  $\beta' := \left( \frac{B}{\sqrt{d/2}}\1_{d/2}, -\frac{E}{\sqrt{d/2}}\1_{d/2} \right)$. If $D$ denotes the distribution of a labeled example $(X,y)$ drawn from the linear regression model of \Cref{def:linear_model} with regressor $\beta$ and standard deviation of additive noise $\sigma$, and $D'$ denotes the distribution of a labeled example according to the linear regression model with regressor $\beta'$ and standard deviation of additive noise $\sigma$, then there exists a coupling $\Pi$ between $D$ and $D'$ such that 
    \begin{align}
        &\E_{((X,y),(X',y')) \sim \Pi}\left[ \sum_{i=1}^d \1(X_i \neq X_i') \right] = O\left( \frac{E}{\sigma} + \frac{\sqrt{d}E}{B}  \right), \label{eq:dissagreements1}\\
        \text{and} \quad &\E_{((X,y),(X',y')) \sim \Pi}\left[  \1(y \neq y') \right] = 0. \label{eq:dissagreements2}
    \end{align}
\end{lemma}

We will prove the lemma at the end. In order to use the conclusion of \Cref{lem:core-coupling-remaining} in combination with \Cref{fact:coupling} we need the right hand sides in \Cref{eq:dissagreements1,eq:dissagreements2} to be at most $\eta d/2$. For this it suffices to assume that $ C E/\sigma < \eta d/2$ and $C \sqrt{d}E  /B < \eta d/2$, where $C$ is the hidden constant in the big-$O$ notation. Equivalently, we need to assume $E < \frac{1}{2C }\min(\eta d \sigma, \eta \sqrt{d} B)$. Also note that our assumption $\eta \leq 1/d$ from the lemma statement also ensures that $E/\sigma<1$, which is required by \Cref{lem:core-coupling-remaining}. Using the value $E := \frac{1}{2C }\min(\eta d \sigma, \eta \sqrt{d} B)$ we can thus apply \Cref{fact:coupling} and obtain that the hypothesis testing defined in the beginning is not solvable with probability better than $\frac{1}{2}(1+(d+1)e^{-\Omega(-\eta n)})$. By the estimation to hypothesis testing reduction, the error lower bound against any estimator is at least a constant multiple of 
\begin{align*}
    \|\beta - \beta' \| &\gtrsim E \\
    &\gtrsim  \min(\eta d \sigma, \eta \sqrt{d} B) \\
    &=   \min \left(\eta d \sigma, \eta \sqrt{d} (B+E)\frac{B}{B+E}\right)\\
    &\gtrsim \min \left(\eta d \sigma, \eta \sqrt{d} (B+E) \right)  \tag{see explanation below}\\
    &\gtrsim  \min(\eta d \sigma, \eta \sqrt{d} \|\beta\|) \tag{by construction $\|\beta\| = \Theta(B+E)$}
\end{align*}

where the fourth line used the following: first, since $\eta \lesssim 1/\sqrt{d}$ we have $E \lesssim \eta \sqrt{d} B \lesssim B$, and then this further implies $B/(B+E) \gtrsim B$.
We now move to showing \Cref{lem:core-coupling-remaining}.
\begin{proof}(Proof of \Cref{lem:core-coupling-remaining})

    The coupling we claim satisfies the guarantee in the lemma statement is the procedure that generates $(X, y)$ and $(X', y')$, which we will describe shortly in pseudocode form. Before presenting the pseudocode of that procedure, we introduce some notation: let $X_1$ denote the first $d/2$ coordinates of $X$, and $X_2$ the second half; we use similar notation, $X_1'$ and $X_2'$, for the corresponding coordinates of $X'$. 
    
    Additionally, for reference purposes only, we will use another pair of labeled examples, $(\tilde X, \tilde y)$ and $(\tilde X', \tilde y')$, drawn from the distributions $D$ and $D'$ defined in the lemma statement (the goal in our generating procedure will be for the output example $(X, y)$ to match the distribution of $(\tilde X, \tilde y)$ and for the second output example $(X', y')$ to match the distribution of $(\tilde X', \tilde y')$). We will also let the notation $\tilde X_1,\tilde X_2$ for denoting the first half and second half of coordinates of $\tilde X$, and similar notation for $\tilde X'$. Finally we will denote $\tilde S_2 := \1_{d/2}^\top \tilde X_2/\sqrt{d/2}$ and $\tilde S_2' := \1_{d/2}^\top \tilde X_2'/\sqrt{d/2}$ (where $\1_{d/2}$ is the all-ones vector of length $d/2$). We will use similar notation $\tilde S_1, \tilde S_1'$ for the scaled sum of the first half of the coordinates. The procedure defining our coupling is as follows. 

    \begin{enumerate}
        \item Draw $t$ from the distribution of $\tilde y$ and set $y=t$ and $y'=t$.
        \item Draw $s_2,s_2'$ from a coupling that couples the distribution of $\tilde S_2$ conditioned on $\tilde y = t$ and the distribution of $\tilde S_2'$ conditioned on $\tilde y' = t$, and also has the additional property that $\Pr[s_2 \neq s_2']= \dtv(\tilde S_2, \tilde S_2')$ and $\E[|s_2 - s_2'|] = |\E[s_2] - \E[s_2']|$. Note that such a coupling exists by \Cref{fact:maximal}.
        \item If $s_2 = s_2'$: 
        \begin{enumerate}
            \item Draw $s_1,s_1'$ using a coupling from \Cref{fact:maximal} that couples the distribution of $\tilde S_1$, conditioned on $\tilde S_2 = s_2$ and $\tilde y = t$ and the distribution of $\tilde S_1'$ conditioned on $\tilde S_2' = s_2$ and $\tilde y' = t$.
            \item If $s_1=s_1'$: Draw $X_1$ from the distribution of $\tilde X_1$ conditioned on $\1_{d/2}^\top \tilde X_1/\sqrt{d/2} = s_1$. Set $X_1' = X_1$. Draw $X_2$ from the distribution of $\tilde X_2$ conditioned on $\1_{d/2}^\top \tilde X_2/\sqrt{d/2} = s_2$. Set $X_2' = X_2$.
            \item If $s_1 \neq s_1'$: Use \Cref{lem:impr-coupling} to couple the distribution  of $\tilde X_1$ conditioned on $\1_{d/2}^\top \tilde X_1/\sqrt{d/2} = s_1$ and the distribution of $\tilde X_1'$ conditioned on $\1_{d/2}^\top \tilde X_1'/\sqrt{d/2} = s_1'$ and draw the pair $(X_1,X_1')$ from that coupling. Draw $X_2$ from the distribution of $\tilde X_2$ conditioned on $\1_{d/2}^\top \tilde X_2/\sqrt{d/2} = s_2$. Set $X_2' = X_2$.
        \end{enumerate}
        \item If $s_2 \neq s_2'$:
        \begin{enumerate}
            \item Draw $s_1,s_1'$ using a coupling that couples the distribution of $\tilde S_1$, conditioned on $\tilde S_2 = s_2$ and $\tilde y = t$ and the distribution of $\tilde S_1'$ conditioned on $\tilde S_2' = s_2'$ and $\tilde y' = t$, and satisfies $\E[|s_1-s_1'|] \leq |\E[s_1] - \E[s_1']|$ (by~\Cref{fact:maximal}).
            \item Use \Cref{lem:impr-coupling} to couple the distribution  of $\tilde X_1$ conditioned on $\1_{d/2}^\top \tilde X_1/\sqrt{d/2} = s_1$ and the distribution of $\tilde X_1'$ conditioned on $\1_{d/2}^\top \tilde X_1'/\sqrt{d/2} = s_1'$ and draw the pair $(X_1,X_1')$ from that coupling.
        \end{enumerate}
        \item Return $(X_1,X_2,y), (X_1',X_2',y')$.
    \end{enumerate}

    \noindent By construction $(X, y)$ follows the same distribution as $(\tilde X, \tilde y)$ (i.e., $D$ from the lemma statement) and $(X', y')$ follows the same distribution as $(\tilde X', \tilde y')$ (i.e., $D'$) thus it is a valid coupling between $D$ and $D'$.  This can be seen by tracking how $(X, y)$ was created (and identical argument will apply to $(X', y')$). The construction samples $(X,y)$ by following the chain rule decomposition of the joint distribution under $D$: we first sample $y$, then sample $S_2$ and $X_2$ conditioned on $y$, and finally sample $S_1$ then $X_1$ conditioned on all previously sampled quantities. Since each step uses the correct conditional distribution under the (null) model, the resulting joint distribution matches that of $(\tilde X, \tilde y)$.

    We now bound the coordinate-wise disagreements. Towards that end, we will need to derive the forms of all the conditional distributions mentioned in the pseudocode. This can be done using the following fact about Gaussians.

    \begin{fact}
        If $\left[\begin{matrix} y_1 \\ y_2 \end{matrix} \right] \sim \cN\left(\left[\begin{matrix} \mu_1 \\ \mu_2 \end{matrix} \right] , \left[ \begin{matrix} \Sigma_{11} &\Sigma_{12} \\ \Sigma_{21} &\Sigma_{22} \end{matrix} \right] \right)$, then $y_1 | y_2 \sim \cN(\bar{\mu}, \bar{\Sigma})$, with $\bar{\mu} = \mu_1 + \Sigma_{12}\Sigma_{22}^{-1}(y_2 - \mu_2)$ and $\bar\Sigma = \Sigma_{11} - \Sigma_{12} \Sigma_{22}^{-1} \Sigma_{21}$.
    \end{fact}

    Using this fact, the conditional distributions of $\tilde S_2$ and $\tilde S_2'$ conditioned on the label values are
    \begin{align} 
        \tilde S_2 | \tilde y = t \sim \cN\left( \frac{E t}{B^2 + E^2 + \sigma^2}, 1- \frac{E^2}{B^2 + E^2 + \sigma^2} \right), \label{label:s2_distribution1}\\
        \tilde S_2' | \tilde y' = t \sim \cN\left( \frac{-E t}{B^2 + E^2 + \sigma^2}, 1- \frac{E^2}{B^2 + E^2 + \sigma^2} \right). \label{label:s2_distribution2}
    \end{align}
    Again, using the fact, the conditional distributions of the sum of the first half of coordinates $\tilde S_1$ and $\tilde S_1'$ given the sums of the second half of coordinates and the label values are:
    \begin{align}
        \tilde S_1 | \tilde S_2 = s_2,  \tilde y = t \sim \cN\left( \frac{B (t-E s_2)}{B^2  + \sigma^2}, 1- \frac{B^2}{B^2  + \sigma^2} \right), \label{label:s1_distribution1}\\
        \tilde S_1' | \tilde S_2' = s_2',\tilde y' = t \sim \cN\left( \frac{B (t+E s_2')}{B^2  + \sigma^2}, 1- \frac{B^2}{B^2  + \sigma^2} \right).\label{label:s1_distribution2}
    \end{align}

    We are now ready to analyze the expected number of disagreeing coordinates.
    We will start with the analysis of the disagreeing coordinates, conditioned on the value of the label being $t$ (and will take expectation over $t$ at the end). We also consider the first half of coordinates and second half of coordinates separately. We will also denote by $\cE_{\mathrm{bad}}$ the event that $s_1 \neq s_1'$ or $s_2 \neq s_2'$.

    For the second half we have the following:
    \begin{align}
        \E &\left[ \sum_{i=1}^{d/2} \1(X_2(i) \neq X_2(i)') \;\middle|\;  y =t\right] \notag\\
        &=\E \left[ \sum_{i=1}^{d/2} \1(X_2(i) \neq X_2(i)') \;\middle|\; \cE_{\mathrm{bad}},y =t\right] \Pr\left[ \cE_{\mathrm{bad}} \mid  y =t \right] \tag{since disagreements occur only if $\cE_{\mathrm{bad}}$ happens} \\
        &\leq \left( 1+\E \left[ |s_2-s_2'| \sqrt{\frac{d}{2}} \;\middle|\; \cE_{\mathrm{bad}},  y =t \right] \right)\Pr\left[ \cE_{\mathrm{bad}} \;\middle|\; y =t \right] \tag{using \Cref{lem:impr-coupling}} \\
        &= \Pr\left[ \cE_{\mathrm{bad}} \;\middle|\; y =t \right]  + \Pr\left[ \cE_{\mathrm{bad}} \;\middle|\;  y =t \right] \E \left[ |s_2-s_2'| \sqrt{\frac{d}{2}} \;\middle|\; \cE_{\mathrm{bad}},  y =t\right] \notag \\
        &= \Pr\left[ \cE_{\mathrm{bad}} \;\middle|\;  y =t \right] + \E \left[ |s_2-s_2'| \sqrt{\frac{d}{2}} \;\middle|\; y =t\right]. \label{eq:final_s2}
    \end{align}
    Regarding the two terms above, first, we have that by \Cref{fact:maximal}, (in particular the last part of the conclusion of that fact)
    \begin{align}
        \E \left[ |s_2-s_2'|  \;\middle|\; y =t\right]  
        &=  \left| \E[s_2] - \E[s_2] \right| = \frac{2E|t|}{B^2 + E^2 + \sigma^2}. \label{eq:diff_s}
    \end{align}
    For the first term in \Cref{eq:final_s2} we will use a union bound to relate it to the probabilities of $s_2\neq s_2'$ and $s_1 \neq s_1'$. Then we note that the former is equal to the total variation distance between the distributions of $s_2$ and $s_2'$ (conditioned on $y=t$), which are the same distributions as the Gaussians in \Cref{label:s2_distribution1,label:s2_distribution2}. Similarly the probability of $s_1 \neq s_1'$ is the total variation distance between the Gaussians in \Cref{label:s1_distribution1,label:s1_distribution2}. Thus,
    \begin{align*}
        \Pr\left[ \cE_{\mathrm{bad}} \;\middle|\;  y =t \right]
        &\leq  \Pr\left[ s_2 \neq s_2' \;\middle|\;  y =t \right] + \Pr\left[ s_1 \neq s_1' \;\middle|\;  y =t \right]\\
        &= \frac{2E|t|}{B^2 + E^2 + \sigma^2} + \frac{E\; B \; \E[|s_2-s_2'| \; \mid \; y=t]}{B^2 + E^2}\\
        &= \frac{2E|t|}{B^2 + E^2 + \sigma^2} + \frac{E\; B }{B^2 + E^2}\frac{2E|t|}{B^2 + E^2 + \sigma^2}.
    \end{align*}
    Taking expectation over $t\sim \cN(0,B^2+E^2+\sigma^2)$ (which is the distribution of the label values in our linear model) in the above inequality we obtain:
    \begin{align}
        \Pr\left[ \cE_{\mathrm{bad}}  \right]
        &\leq \frac{2E}{\sqrt{B^2 + E^2 + \sigma^2} } + \frac{E\; B }{B^2 + E^2}\frac{2E}{\sqrt{B^2 + E^2 + \sigma^2}} \lesssim \frac{E}{\sigma},\label{eq:prob_bad_event}
    \end{align}
    where the last step uses that $E\lesssim B$, $\sigma,E,B\geq 0$.

    Combining \Cref{eq:final_s2,eq:prob_bad_event,eq:diff_s} and taking expectation over the label value $t \sim \cN(0,B^2+E^2+\sigma^2)$ for the terms that we have not done it already,  we have that the expected number of disagreeing coordinates in the second half of the vector is
    \begin{align*}
        \E \left[ \sum_{i=1}^{d/2} \1(X_2(i) \neq X_2(i)') \right]
        &\lesssim \frac{E}{\sqrt{B^2 + E^2 + \sigma^2}} + \frac{\sqrt{d}E}{\sqrt{B^2 + E^2 + \sigma^2}} \leq\frac{E}{\sigma} + \frac{\sqrt{d}E}{B}.
    \end{align*}

    We now work similarly in order to bound the expected number of disagreements in the first half of the coordinates. First, conditioned on $t$ we have that
    \begin{align}
        \E &\left[ \sum_{i=1}^{d/2} \1(X_1(i) \neq X_1(i)') \;\middle|\;  y =t\right] \notag \\
        &=\E \left[ \sum_{i=1}^{d/2} \1(X_1(i) \neq X_1(i)') \;\middle|\;  \cE_{\mathrm{bad}},y =t\right] \Pr\left[ \cE_{\mathrm{bad}} \mid  y =t \right] \tag{since disagreements occur only during $\cE_{\mathrm{bad}}$} \\
        &\leq \left( 1+\E \left[ |s_1-s_1'| \sqrt{\frac{d}{2}} \;\middle|\; \cE_{\mathrm{bad}},  y =t \right] \right)\Pr\left[ \cE_{\mathrm{bad}} \;\middle|\; y =t \right] \tag{using \Cref{lem:impr-coupling}} \\
        &= \Pr\left[ \cE_{\mathrm{bad}} \;\middle|\; y =t \right]  + \Pr\left[ \cE_{\mathrm{bad}} \;\middle|\;  y =t \right] \E \left[ |s_1-s_1'| \sqrt{\frac{d}{2}} \;\middle|\; \cE_{\mathrm{bad}},  y =t\right] \notag \\
        &= \Pr\left[ \cE_{\mathrm{bad}} \;\middle|\;  y =t \right] + \E \left[ |s_1-s_1'| \sqrt{\frac{d}{2}} \;\middle|\; y =t\right] . \label{eq:final_s1}
    \end{align}
    The probability term (after taking expectation over $t$) has been bounded in \Cref{eq:prob_bad_event}. For the expectation term, we can again take into consideration that $s_1,s_1'$ have been generated from the coupling of \Cref{fact:maximal}, and because the means of the distributions are as shown in \Cref{label:s1_distribution1,label:s1_distribution2}, we have
    \begin{align}
        \E \left[ |s_1-s_1'| \;\middle|\; y =t\right] 
        = \E \left[ \frac{B\; E \;|s_2 - s_2'|}{B^2 + \sigma^2} \;\middle|\; y =t\right]  
        = \frac{B E}{B^2 + \sigma^2} \frac{2E|t|}{B^2+E^2+\sigma^2}. \label{eq:s1_second_term}
    \end{align}
    Combining \Cref{eq:prob_bad_event,eq:final_s1,eq:s1_second_term} and taking expectation over $t \sim \cN(0,B^2+E^2+\sigma^2)$ we have that
    \begin{align*}
        \E &\left[ \sum_{i=1}^{d/2} \1(X_1(i) \neq X_1(i)') \right]
        \lesssim \frac{E}{\sigma} + \frac{\sqrt{d}B E}{B^2 + \sigma^2} \frac{E\;\E_{t\sim \cN(0,B^2+E^2+\sigma^2)}[|t|]}{B^2+E^2+\sigma^2} \\
        &\lesssim \frac{E}{\sigma} + \frac{\sqrt{d}B E^2}{(B^2 + \sigma^2)\sqrt{B^2+E^2+\sigma^2}} 
        \lesssim \frac{E}{\sigma} + \frac{\sqrt{d}E}{B}.
    \end{align*}
    This completes the proof of \Cref{lem:core-coupling-remaining}.

\end{proof}
The proof of \Cref{thm:small-eta} is now complete.
\end{proof}

\end{document}